\documentclass[11pt]{article}

\usepackage[utf8]{inputenc}
\usepackage[T1]{fontenc}
\usepackage[sc]{mathpazo}
\usepackage{microtype}
\usepackage{amsmath}
\usepackage{amssymb}
\usepackage{amsthm}
\usepackage{bm}
\usepackage{dsfont}
\usepackage{authblk}
\usepackage{fullpage}
\usepackage{comment}
\usepackage{tikz}
\usepackage{mathtools}
\usepackage[shortlabels]{enumitem}
\usepackage{complexity}
\usepackage[backend=biber, style=alphabetic, maxbibnames=99, url=false]{biblatex}
\usepackage{hyperref}
\usepackage[capitalize, nameinlink]{cleveref}
\usepackage{nag}

\newcommand{\declarecolor}[2]{\definecolor{#1}{RGB}{#2}\expandafter\newcommand\csname #1\endcsname[1]{\textcolor{#1}{##1}}}
\declarecolor{White}{255, 255, 255}
\declarecolor{Black}{0, 0, 0}
\declarecolor{LightGray}{216, 216, 216}
\declarecolor{Gray}{127, 127, 127}
\declarecolor{Orange}{237, 125, 49}
\declarecolor{LightOrange}{251,229, 214}
\declarecolor{Yellow}{255, 192, 0}
\declarecolor{LightYellow}{255, 242, 200}
\declarecolor{Blue}{91, 155, 213}
\declarecolor{LightBlue}{222, 235, 247}
\declarecolor{Green}{112, 173, 71}
\declarecolor{LightGreen}{226, 240, 217}
\declarecolor{Navy}{68, 114, 196}
\declarecolor{LightNavy}{218, 227, 243}

\hypersetup{
	colorlinks=true,
	pdfpagemode=UseNone,
	citecolor=Green,
	linkcolor=Navy,
	urlcolor=Black,
	pdfstartview=FitW
}

\newcommand{\declareperson}[1]{\expandafter\newcommand\csname#1\endcsname[1]{\textcolor{Orange}{#1: ##1}}}

\declareperson{Nima}
\declareperson{Shayan}
\declareperson{Cynthia}

\theoremstyle{plain}
\newtheorem{theorem}{Theorem}[section]
\newtheorem{lemma}[theorem]{Lemma}
\newtheorem{corollary}[theorem]{Corollary}
\newtheorem{proposition}[theorem]{Proposition}
\newtheorem{fact}[theorem]{Fact}
\newtheorem{conjecture}[theorem]{Conjecture}

\theoremstyle{definition}

\theoremstyle{remark}
\newtheorem{remark}[theorem]{Remark}
\newtheorem{example}[theorem]{Example}

\newlist{parts}{enumerate}{10}
\setlist[parts]{label=\arabic*.,ref=\arabic*}
\makeatletter
\if@cref@capitalise
	\crefname{partsi}{Part}{Parts}
\else
	\crefname{partsi}{part}{parts}
\fi
\makeatother
\Crefname{partsi}{Part}{Parts}

\usetikzlibrary{shapes, patterns, decorations, fit, intersections}

\newcommand*{\Q}{{\mathbb{Q}}}

\newcommand*{\Z}{{\mathbb{Z}}}
\let\R\relax
\newcommand*{\R}{{\mathbb{R}}}

\newcommand*{\I}{{\mathcal{I}}}
\newcommand*{\B}{{\mathcal{B}}}
\let\P\relax
\newcommand*{\P}{{\mathcal{P}}}
\let\H\relax
\newcommand*{\H}{{\mathcal{H}}}
\newcommand*{\1}{{\mathds{1}}}

\newcommand*{\bmlambda}{{\bm{\lambda}}}
\newcommand*{\bmkappa}{{\bm{\kappa}}}

\let\A\relax
\newcommand*{\A}{{\bm{A}}}
\newcommand*{\V}{{\bm{V}}}
\let\a\relax
\newcommand*{\a}{{\bm{a}}}
\let\b\relax
\newcommand*{\b}{{\bm{b}}}
\let\p\relax
\newcommand*{\p}{{\bm{p}}}
\let\v\relax
\newcommand*{\v}{{\bm{v}}}
\newcommand*{\w}{{\bm{w}}}
\newcommand*{\y}{{\bm{y}}}
\newcommand*{\z}{{\bm{z}}}

\newcommand*{\ef}[2]{#1\star#2}
\newcommand*{\barK}{\overline{K}}
\newcommand*{\barM}{M_0}
\newcommand*{\ALG}{\beta}
\newcommand*{\Hess}[1]{\nabla^2 #1}

\newcommand*{\tmu}{{\tilde{\mu}}}
\newcommand*{\tM}{{\tilde{M}}}
\newcommand*{\tN}{{\tilde{N}}}
\newcommand*{\tB}{{\tilde{B}}}
\newcommand*{\hats}{{\hat{s}}}
\newcommand*{\hatt}{{\hat{t}}}
\newcommand*{\haty}{{\hat{y}}}
\newcommand*{\hatz}{{\hat{z}}}
\newcommand*{\hatbmy}{{\bm{\hat{y}}}}
\newcommand*{\hatbmz}{{\bm{\hat{z}}}}

\newcommand*{\SharpP}{\ComplexityFont{\#P}}
\let\poly\relax
\DeclareMathOperator{\poly}{poly}
\DeclareMathOperator{\supp}{supp}
\DeclareMathOperator{\conv}{conv}
\DeclareMathOperator{\rank}{rank}
\DeclareMathOperator*{\argmax}{argmax}
\let\Im\relax
\DeclareMathOperator{\Im}{Im}

\providecommand{\given}{}
\DeclarePairedDelimiterX{\card}[1]{\lvert}{\rvert}{\renewcommand\given{\nonscript\:\delimsize\vert\nonscript\:\mathopen{}}#1}
\DeclarePairedDelimiterX{\abs}[1]{\lvert}{\rvert}{\renewcommand\given{\nonscript\:\delimsize\vert\nonscript\:\mathopen{}}#1}
\DeclarePairedDelimiterX{\norm}[1]{\lVert}{\rVert}{\renewcommand\given{\nonscript\:\delimsize\vert\nonscript\:\mathopen{}}#1}
\DeclarePairedDelimiterX{\tuple}[1]{\lparen}{\rparen}{\renewcommand\given{\nonscript\:\delimsize\vert\nonscript\:\mathopen{}}#1}
\DeclarePairedDelimiterX{\parens}[1]{\lparen}{\rparen}{\renewcommand\given{\nonscript\:\delimsize\vert\nonscript\:\mathopen{}}#1}
\DeclarePairedDelimiterX{\brackets}[1]{\lbrack}{\rbrack}{\renewcommand\given{\nonscript\:\delimsize\vert\nonscript\:\mathopen{}}#1}
\DeclarePairedDelimiterX{\set}[1]\{\}{\renewcommand\given{\nonscript\:\delimsize\vert\nonscript\:\mathopen{}}#1}
\let\Pr\relax
\DeclarePairedDelimiterXPP{\Pr}[1]{\mathbb{P}}[]{}{\renewcommand\given{\nonscript\:\delimsize\vert\nonscript\:\mathopen{}}#1}
\DeclarePairedDelimiterXPP{\PrX}[2]{\mathbb{P}_{#1}}[]{}{\renewcommand\given{\nonscript\:\delimsize\vert\nonscript\:\mathopen{}}#2}
\DeclarePairedDelimiterXPP{\Ex}[1]{\mathbb{E}}[]{}{\renewcommand\given{\nonscript\:\delimsize\vert\nonscript\:\mathopen{}}#1}
\DeclarePairedDelimiterXPP{\ExX}[2]{\mathbb{E}_{#1}}[]{}{\renewcommand\given{\nonscript\:\delimsize\vert\nonscript\:\mathopen{}}#2}
\DeclarePairedDelimiter{\dotprod}{\langle}{\rangle}
\newcommand*{\eval}[1]{\left.#1\right\rvert}

\title{Log-Concave Polynomials I: Entropy and a Deterministic Approximation Algorithm for Counting Bases of Matroids}

\author{Nima Anari}
\affil{\small Computer Science Department\\ Stanford University, \textsf{anari@cs.stanford.edu}}

\author{Shayan Oveis Gharan}
\affil{\small Computer Science and Engineering\\ University of Washington, \textsf{shayan@cs.washington.edu}}

\author{Cynthia Vinzant}
\affil{\small Department of Mathematics\\ North Carolina State University, \textsf{clvinzan@ncsu.edu}}

\bibliography{refs}

\begin{document}
	\maketitle
	
	\begin{abstract}
	We give a deterministic polynomial time $2^{O(r)}$-approximation algorithm for the number of bases of a given matroid of rank $r$ and the number of common bases of any two matroids of rank $r$. To the best of our knowledge, this is the first nontrivial deterministic approximation algorithm that works for arbitrary matroids. Based on a lower bound of \textcite{ABF94} this is almost the best possible result assuming oracle access to independent sets of the matroid. 

	There are two main ingredients in our result: For the first, we build upon recent results of \textcite{AHK15,HW17} on combinatorial hodge theory to derive a connection between matroids and log-concave polynomials. We expect that several new applications in approximation algorithms will be derived from this connection in future. Formally, we prove that the multivariate generating polynomial of the bases of any matroid is log-concave as a function over the positive orthant. For the second ingredient, we develop a general framework for approximate counting in discrete problems, based on convex optimization. The connection goes through subadditivity of the entropy. For matroids, we prove that an approximate superadditivity of the entropy holds by relying on the log-concavity of the corresponding polynomials. 

	\end{abstract}
	
	
	\section{Introduction}
\label{sec:introduction}

Efficient algorithms for optimizing linear functions over convex sets, i.e., convex programming, are one of the pinnacles of algorithm design. Convex sets yield easy instances of optimization in the continuous world. Much the same way, matroids yield easy instances of optimization in the discrete world. Going beyond optimization, computing the volume of or sampling from convex sets is well understood algorithmically; however there has not been an analogous progress on counting problems involving matroids. In this work, we try to address this issue by designing nearly tight deterministic approximate counting algorithms for discrete structures involving matroids and their intersections. We introduce a general optimization-based algorithm for approximate counting involving discrete objects, and show that our algorithm performs well for matroids and their intersections.

A matroid $M = (E,\I)$ is a structure consisting of a finite ground set $E$ and a non-empty collection $\I$ of \emph{independent} subsets of $E$ satisfying: 
\begin{enumerate}[i)]  
	\item If $S \subseteq T$ and $T \in \I$, then $S \in \I$.  
	\item If $S,T\in \I$ and $\card{T} > \card{S}$, then there exists an element $i \in T\setminus S$ such that $S\cup \set{i} \in \I$. 
\end{enumerate}
The \emph{rank} of a matroid is the size of the largest independent set of that matroid. If $M$ has rank $r$, any set $S\in \I$ of size $r$ is called a \emph{basis} of $M$. Let $\B_M \subset \I$ denote the set of bases of $M$.

Many optimization problems are well understood on matroids. Matroids are exactly the class of objects for which an analogue of Kruskal's algorithm works and gives the smallest weight basis. 

One can associate to any matroid $M$ a polytope $\P_M$, defined by exponentially many constraints, called the \emph{matroid base polytope}. The vertices of $\P_M$ are the indicator vectors of all bases of $M$, i.e., $\P_M=\conv\set{\1_B\given B\in \B_M}$. Furthermore, using the duality of optimization and separation, one can design a separation oracle for $\P_M$ in order to minimize any convex function over $\P_M$ \cite{Cun84}.

More difficult problems associated to matroids come from counting. For example, given a matroid $M$, is there a polynomial time algorithm to count the number of bases of $M$? This problem is \SharpP-hard in the worst case even if the matroid is representable over a finite field \cite{Sno12}, so the next natural question is: How well can we approximate the number of bases of a given matroid $M$ in polynomial time? This is the main question addressed in this paper. Note that the number of bases of any matroid $M$ of rank $r$ is at most $\binom{\card{E}}{r}\approx \card{E}^r$, so there is a simple $\card{E}^r$ approximation to the number of bases of $M$.

We also address counting problems on the intersection of two matroids. Given two matroids $M=(E,\I_M),N=(E,\I_N)$ of rank $r$ on the same ground set $E$, the matroid intersection problem is to optimize a (linear) function over all bases $B$ common to both $M$ and $N$. This problem can also be solved in polynomial time because $\P_M \cap \P_N$ is exactly the convex hull of the indicator vectors of $\B_M\cap \B_N$ \cite[see, e.g.,][]{Sch03}. Perhaps, the most famous example of matroid intersection  is perfect matchings in bipartite graphs. Since we can optimize over the intersection of two matroids, it is natural to ask if one can approximate the number of bases common to two rank-$r$ matroids $M$ and $N$. This is the second problem that we address in this paper. 

Note that there are \NP-hard problems involving the intersection of just three matroids, e.g., the Hamiltonian path problem. It is \NP-hard to test if there is a single basis in the intersection of three matroids. We will not discuss intersections of more than two matroids in this paper, since any multiplicative approximation would be \NP-hard.

\subsection{Previous Work}

There is a long line of work on designing approximation algorithms to count the number of bases of a matroid. One idea is to use the Markov Chain Monte Carlo technique. For any matroid $M$, there is a well-known chain, called the basis exchange walk, which mixes to the uniform distribution over all bases. Mihail and Vazirani conjectured, about three decades ago, that the chain mixes in polynomial time, and hence one can approximate the number of bases of any matroid on $n$ elements within $1+\epsilon$ factor in time $\poly(n, 1/\epsilon)$. To this date, this conjecture has been proved only for a special class of matroids known as \emph{balanced matroids} \cite{MS91,FM92}. Balanced matroids are special classes of matroids where the uniform distribution over the bases of the matroid, and any of its minors, satisfies the pairwise negative correlation property. Unfortunately, many interesting matroids are not balanced. An important example is the  matroid of all acyclic subsets of edges of a graph $G=(V,E)$ of size at most $k$ (for some $k<\card{V}-1$) \cite{FM92}.

Many of the extensive results in this area \cite{Gam99,JS02,JSTV04,Jur06,Clo10,CTY15,AOR16} only study approximation algorithms for a limited class of matroids, and not much is known beyond the class of balanced matroids. 

Most of the classical results in approximate counting rely on randomized algorithms based on the Markov Chain Monte Carlo technique. There are also a few results in the literature that exploit tools from convex optimization \cite{Bar97,HSKK96,BS07}. To the best of our knowledge, the only non-trivial approximation algorithm that works for any matroid is a \emph{randomized} algorithm of \textcite{BS07} that gives, roughly, a $\log(\card{E})^r$ approximation factor to the number of bases of a  given matroid of rank $r$ and the number of common bases of any two given matroids of rank $r$, in the worst case. We remark that this algorithm works for any family of subsets, not just matroids and their intersections, assuming access to an optimization oracle. The approximation factor gets better if the number of bases of the given matroid(s) are significantly less than $\binom{\card{E}}{r}\simeq \card{E}^r$.

On the negative side, \textcite{ABF94} showed that any deterministic polynomial time algorithm that has access to the matroid $M$ on $n$ elements only through an independence oracle can only approximate the number of bases of $M$ up to a factor of $2^{O(n/\log(n)^2)}$. They actually showed the stronger result that any deterministic algorithm making $k$ queries to the independence oracle must have an approximation factor of at least $2^{\Omega(n/\log(k)^2)}$, as long as $k=2^{o(n)}$. An immediate corollary is a \emph{rank-dependent} lower bound, namely that any deterministic algorithm making polynomially many independence queries to a matroid of rank $r$ must have an approximation ratio of $2^{\Omega(r/\log(n)^2)}$ as long as $r\gg \log(n)$. This is because we can always start with a matroid on $\simeq r$ elements and add loops to get a matroid on $n$ elements without changing the number of bases or the rank.

The problem of approximating the number of bases in the intersection of two matroids $M,N$ is very poorly understood. \Textcite{JSV04} give a randomized polynomial time approximation to the number of perfect matchings of a bipartite graph, a special case of intersection of two matroids, up to a factor of $1+\epsilon$. As for deterministic algorithms, for this special case of bipartite perfect matchings, a $2^{O(r)}$-approximation was first introduced by \textcite{LSW98}, relying on the Van der Waerden conjecture, and later an improvement in the base of the exponent was achieved by \textcite{GS14}. Recently, a subset of the authors \cite{AO17} have shown that one can approximate the number of bases in the intersection of two \emph{real stable} matroids, having oracle access to each of their generating polynomials, up to a $2^{O(r)}$ multiplicative error \cite[also, cf.][]{SV17}. Real stable matroids can be seen as a special case of balanced matroids. See the techniques in \cref{sec:matroidintersection} and \textcite{AO17} for more details. 

\subsection{Our Results}
The main result of this paper is the following.  
\begin{theorem}\label{thm:mainmatroid}
	Let $M=([n],\I)$ be an arbitrary matroid of rank $r$ given by an independence oracle; that is for every $S\subseteq [n]$, one can query the oracle if $S\in \I$. There is a \emph{deterministic} polynomial time algorithm that outputs a number $\ALG$ satisfying
	\[ \max\set{2^{-O(r)}\ALG, \sqrt{\ALG}}\leq \card{\B_M}\leq \ALG, \]
	where $\B_M$ is the set of bases of $M$.
\end{theorem}

Our algorithm can be implemented with only oracle access to the independent sets of the matroid. Therefore, by the work of \textcite{ABF94}, this is almost the best we can hope for any deterministic algorithm.

As an immediate corollary of the above result we can count the number of independent sets of any given size $k$. This is because independent sets of size $k$ form the bases of the truncation of the original matroid, which itself is a matroid.
\begin{corollary}
	Let $M$ be an arbitrary matroid given by an independence oracle. There is a deterministic polynomial time algorithm that for any given integer $k$ outputs a number $\ALG$ such that
	\[ \max\set{2^{-O(k)}\ALG,\sqrt{\ALG}} \leq \card{\I^k_M} \leq \ALG, \]
	where $\I^k_M$ is the set of independence sets of $M$ of size exactly $k$.
\end{corollary}

Building further on our techniques, we show that one can approximate the number of bases in the intersection of \emph{any} two matroids. 
\begin{theorem}\label{thm:matroidintersection}
	Let $M$ and $N$ be two matroids of rank $r$ on the ground set $[n]$ given by independence oracle. There is a \emph{deterministic} polynomial time algorithm that outputs a number $\ALG$ such that
	\[ 2^{-O(r)}\ALG\leq \card{\B_M\cap \B_N} \leq  \ALG,\]
	where $\B_M, \B_N$ are the sets of bases of $M$ and $N$, respectively.
\end{theorem}

Counting common bases of two matroids is a self-reducible problem. Roughly speaking, this means that if we want to count common bases that include some given elements $i_1,\dots,i_k\in [n]$ and exclude some other elements $j_1,\dots,j_l\in [n]$, then we get an instance of the same problem; we just have to replace the input matroids by their minors obtained by contracting $i_1,\dots,i_k$ and deleting $j_1,\dots,j_l$.

\Textcite{JS89} showed that for self-reducible counting problems, one can boost the approximation factor of any algorithm at the expense of an increase in running time and using randomization. In particular, as a corollary of \cref{thm:matroidintersection} and the results of \textcite{JS89} we get the following.

\begin{corollary}
	Let $M$ and $N$ be two matroids of rank $r$ on the ground set $[n]$ given by independence oracles. There is a \emph{randomized} algorithm that for any desired $\epsilon,\delta>0$ outputs a number $\beta$ approximating the number of common bases of $M$ and $N$ within a factor of $1-\epsilon$ with probability at least $1-\delta$:
	\[ \Pr[\Big]{(1-\epsilon)\beta \leq \card{\B_M\cap\B_N}\leq \beta}\geq 1-\delta. \]
	The running time of this algorithm is $2^{O(r)}\poly(n, \frac{1}{\epsilon},\log\frac{1}{\delta})$.
\end{corollary}

Counting bases of a single matroid is a special case obtained by letting $M=N$; so this result applies to counting bases of a single matroid as well. Also observe that this algorithm becomes a fully polynomial time randomized approximation scheme (FPRAS) when $r=O(\log n)$.
We also prove that a slight generalization of our algorithmic framework provides $2^{O(r)}$-approximation to weighted counts of bases.
\begin{theorem}
\label{thm:weighted}
Let $M$ and $N$ be two matroids of rank $r$ on the ground set $[n]$ given by independence oracles. There is a deterministic polynomial time algorithm that for any given weights $\bmlambda=(\lambda_1,\dots,\lambda_n)\in \R_{\geq 0}^n$, outputs a number $\ALG$ that $2^{O(r)}$-approximates the $\bmlambda$-weighted intersection of $M$ and $N$:
	\[ 2^{-O(r)}\ALG\leq \sum_{B\in \B_M\cap \B_N}\prod_{i\in B} \lambda_i \leq \ALG. \]
\end{theorem}

\subsection{Techniques}
In this section we discuss the main ideas of our proof. We rely heavily on the \emph{basis generating polynomial} of a matroid $M$, 
\begin{equation}
	\label{eq:gM}
	g_M(z_1,\dots,z_n)=\sum_{B\in \B_M} \prod_{i\in B} z_i.
\end{equation}
For some matroids, such as partition matroids and graphic matroids, the polynomial $g_M$ has a special property called real stability. A multivariate polynomial $g\in\R[z_1,\dots,z_n]$ is \emph{real stable} if $g(z_1, \dots, z_n)\neq 0$ whenever $\Im(z_i)>0$ for all $i=1, \dots, n$. Recently, real stable polynomials have been used for numerous counting and sampling tasks \cite{NS16,AOR16,AO17,SV17,AOSS17,AMOV18}.

Matroids with real stable basis generating polynomials are called \emph{real stable matroids}. Many properties of these matroids can be derived from the stability of their generating polynomials. If $g_M$ is real stable, then the uniform distribution over the bases of the matroid, as well as its minors, satisfy pairwise negative correlation. Then $M$ is a balanced matroid and one can count the number of bases of $M$ within a $1+\epsilon$ multiplicative error by a polynomial time randomized algorithm. 

On the counting side, roughly speaking, real stable matroids are almost all we know how to handle. However, it is known that many matroids are not real stable. Even if we allow arbitrary positive coefficients in front of the monomials in the generating polynomial, instead of uniform coefficients, for some matroids we can never get a real stable polynomial \cite{Bra07}. Here we define a more general class of polynomials, namely log-concave and completely log-concave polynomials, to be able to study all matroids with analytical techniques.

Given a real stable polynomial $g(z_1,z_2,\dots, z_n)$, its univariate restriction $g(z,z,\dots,z)$ is real-rooted, and it follows that its coefficients form a \emph{log-concave} sequence. Recently, \textcite{AHK15} proved that certain univariate polynomials associated with matroids have log-concave coefficients, for \emph{any} matroid. Their work resolved several long standing open problems in combinatorics. It is important to note that such results are very unlikely to follow from the theory of real stability because not all matroids support real stable polynomials.

The first ingredient of our paper is that we exploit some of the results and theory developed by \textcite{AHK15} to show that $\log(g_M(z_1,\dots,z_n))$ is concave as a function on the positive orthant (see \cref{thm:gMlogconcave}). Any real stable polynomial with nonnegative coefficients is log-concave on the positive orthant but the converse is not necessarily true. In this paper we study properties of log-concave polynomials with nonnegative coefficients  as a generalization of real stable polynomials and show that they satisfy many of the closure properties of real stable polynomials.  Here, we mainly focus on  applications in approximate counting, but we hope that these techniques can be used for many other applications in algorithm design, operator theory, and combinatorics.

Our second ingredient is a general framework for approximate counting based on convex optimization. 
We consider probability distributions $\mu:2^{[n]}\to\R_{\geq 0}$ on subsets of $[n]$. 
Firstly, we show that if $\mu$ has a log-concave generating polynomial, we can approximate its entropy using the marginal probabilities of the underlying elements (see \cref{thm:mainentropy}). The marginal probability $\mu_i$ of an element $i\in [n]$ is the probability that $i$ is included in a random sample of $\mu$. We show that $\sum_{i=1}^n \parens{\mu_i \log\frac1{\mu_i}+(1-\mu_i)\log\frac{1}{1-\mu_i}}$ gives a ``good'' approximation of the entropy $\H(\mu)$ of $\mu$. 

This is particularly interesting when $\mu$ is the uniform distribution over the bases of a matroid $M = ([n],\I)$, in which case $\H(\mu)$ equals $\log(\card{\B_M})$. From the marginals $\mu_i$, one can approximate $\H(\mu)$, but finding the marginal probabilities is no easier than estimating $\H(\mu)$. Instead, we observe that the vector of marginals $(\mu_1,\dots,\mu_n)$ must lie in $\P_M$. So instead of trying to find $\mu_i$'s, we use a convex program to find a point $\p=(p_1,\dots,p_n)\in \P_M$ maximizing the sum of marginal entropies, $\sum_{i=1}^n \parens{p_i\log\frac1{p_i}+(1-p_i)\log\frac1{1-p_i}}$. Using properties of maximum entropy convex programs (see \cref{thm:maxentropy}), we show that this also gives a ``good'' approximation of $\H(\mu)=\log(\card{\B_M})$. 

To prove \cref{thm:matroidintersection}, we exploit some of the previous tools that a subset of authors developed \cite{AO17} to approximate the number of bases in the intersection of two real stable matroids. In this paper we generalize these techniques to all matroids. In the process we show that for any matroid $M$ the polynomial $g_M$ is \emph{completely log-concave}, meaning that taking directional derivatives of the polynomial $g_M$ with respect to directions in $\R_{\geq 0}^n$ results in a polynomial that as a function is log-concave on the positive orthant (see \cref{thm:gMlogconcave}).

\subsection{Algorithmic Framework}
\label{sec:framework}
Our algorithms in the cases of a single matroid and the intersection of two matroids are actually instantiations of the same framework that could be applied to more general discrete structures. We use a general framework based on convex programming. Take an arbitrary family $\B\subseteq 2^{[n]}$ of subsets of $\set{1,\dots,n}$ as our discrete object. For us $\B$ will be either the set of bases of a matroid, or the common bases of two matroids, but our framework could be applied to more general families. We assume that we can optimize linear functions over the polytope $\P_\B=\conv\set{\1_B\given B\in \B}$. This is true in both cases involving matroid(s), as long as we have access to the corresponding independence oracle(s).

The key observation is that the entropy of the uniform distribution $\mu$ over the elements of $\B$ equals $\log(\card{\B})$ and that using subadditivity of the entropy we can relate this entropy to the points in  $\P_\B$. More precisely if $\mu_i$ is the marginal probability of element $i$ being in a randomly chosen element of $\B$, then $\tuple{\mu_1,\dots,\mu_n}$ is a point of $\P_\B$; it is precisely the average of all vertices of $\P_\B$. This fact, together with subadditivity of the entropy gives us
\begin{align*} \log(\card{\B})=\H(\mu)&\leq \sum_{i=1}^n\parens*{\mu_i\log\frac{1}{\mu_i}+(1-\mu_i)\log\frac{1}{1-\mu_i}}\\ &\leq \max\set*{\sum_{i=1}^n \parens*{p_i\log\frac{1}{p_i}+(1-p_i)\log\frac{1}{1-p_i}}\given \p=(p_1,\dots,p_n) \in \P_\B}. \end{align*}
The last quantity is something that we can efficiently compute, because we can optimize over the polytope $\P_\B$. Therefore, we have a convex-programming-based way of obtaining an upper bound for $\log(\card{\B})$. By exponentiating the result, we get an upper bound on $\card{\B}$; this is the output of our algorithm, $\ALG$ in \cref{thm:mainmatroid,thm:matroidintersection}.  Our results show that for matroids and intersections of two matroids, this upper bound becomes a ``good'' approximation, in the sense that there is a complementary lower bound. We leave the question of finding more discrete structures for which this algorithm provides a ``good'' approximation open. As a potential direction, we make the following  concrete conjecture.
\begin{conjecture}
	\label{conj:matching}
	Let $G=(V,E)$ be a graph with an even number of nodes, and let $\B\subseteq 2^E$ be the set of all perfect matchings in $G$. Then
	\[ \max\set*{\sum_{i\in E} \parens*{p_i\log\frac{1}{p_i}+(1-p_i)\log\frac{1}{1-p_i}}\given \p \in \P_\B}-O(\card{V})\leq \log(\card{\B}). \]
\end{conjecture}
If \cref{conj:matching} is correct, we immediately get a deterministic polynomial time $2^{O(\card{V})}$-approximation algorithm for counting perfect matchings, because we have efficient optimization over $\P_\B$. To the best of our knowledge, for nonbipartite graphs, no such result is known as of this writing.

Finally we remark that this framework can also handle weighted counting. For weight vector $\bmlambda=(\lambda_1,\dots,\lambda_n)\in \R_{\geq 0}^n$, the $\bmlambda$-weighted count of a family $\B\subseteq 2^{[n]}$ is simply
\[ \sum_{B\in \B}\prod_{i\in B} \lambda_i. \]
For $\lambda_1,\dots,\lambda_n=1$, this quantity becomes $\card{\B}$. To handle $\bmlambda$-weighted counts we can simply change our concave program slightly to
\[ \max\set*{\sum_{i=1}^n\parens*{p_i\log \frac{\lambda_i}{p_i}+(1-p_i)\log \frac{1}{1-p_i}}\given \p=(p_1,\dots,p_n)\in \P_\B}. \]
In \cref{sec:weighted} we show that this is always an upper bound on the logarithm of $\bmlambda$-weighted count of $\B$, and that for matroids and intersections of two matroids, this is a ``good'' approximation.

\subsection{Organization}

The rest of the paper is organized as follows. In \cref{sec:preliminaries}, we go over the necessary preliminaries from matroid theory, convex entropy programs, linear algebra and develop some of the theory of log-concave polynomials. In \cref{sec:hodge} we review some tools from combinatorial Hodge theory and derive a connection with the application of differential operators to the basis generating polynomial of a matroid. Then, in \cref{sec:logconcavity} we prove that the generating polynomial of \emph{any} matroid is log-concave. In \cref{sec:entropylogconcave} we prove that one can approximate the entropy of a log-concave distribution from its marginals. In \cref{sec:maxentropymatroidbases} we prove \cref{thm:mainmatroid} and in \cref{sec:matroidintersection} we prove \cref{thm:matroidintersection}. Finally, \cref{sec:weighted} contains the proof of \cref{thm:weighted}.

\subsection{Acknowledgements} We would like to thank Matt Baker, June Huh, and Josephine Yu for their comments on an early draft of this paper. Part of this work was done while the first and third authors were visiting the Simons Institute for the Theory of Computing. It was partially supported by the DIMACS/Simons Collaboration on Bridging Continuous and Discrete Optimization through NSF grant CCF-1740425. Shayan Oveis Gharan is supported by the NSF grant CCF-1552097 and ONR-YIP grant N00014-17-1-2429. Cynthia Vinzant was partially supported by the National Science Foundation grant DMS-1620014.

	\section{Preliminaries}

\label{sec:preliminaries}

First, let us establish some notational conventions. Unless otherwise specified, all $\log$s are in base $e$. 
We often use bold letters to emphasize symbols representing a vector, array, or matrix of numbers or variables.  All vectors are assumed to be column vectors. For two vectors $\v, \w\in \R^n$, we denote the standard dot product between $\v$ and $\w$ by $\dotprod{\v, \w}=\v^\intercal \w$. We use $\R_{>0}$ and $\R_{\geq 0}$ to denote the set of positive and nonnegative real numbers, respectively, and  $[n]$ to denote $\set{1,\dots,n}$. When $n$ is clear from context, for a set $S\subseteq [n]$, we let $\1_S\in \R^n$ denote the indicator vector of $S$ with $(\1_S)_i=1$ if $i\in S$, and is $0$ otherwise. Abusing notation, we let $\1_i=\1_{\set{i}}$ be the $i$-th element of the standard basis for $\R^n$. For vectors $\z, \p\in \R^n$ we use $\z^\p$ to denote $\prod_{i=1}^n z_i^{p_i}$. Similarly we let $e^\p$ denote $\prod_{i=1}^n e^{p_i}=e^{\sum_{i=1}^n p_i}$. For a vector $\z\in \R^n$ and a set $S\subseteq [n]$, we let $\z^S$ denote $\prod_{i\in S} z_i$.

We use $\partial_{z_i}$ or $\partial_i$ to denote the partial differential operator $\partial/\partial z_i$. Given $\v=(v_1,\dots,v_n)\in \R^n$, we will use $\partial_{\v}$ to denote the differential operator $\sum_{i=1}^n v_i \partial_i$. For a collection of vectors $\v_1,\dots,\v_k\in \R^n$, we use $D_{\v_1,\dots,\v_k}$ to denote $\prod_{i=1}^k \partial_{\v_i}$. For a matrix $\V=\brackets{\v_1\given \dots\given \v_k}\in \R^{n\times k}$, viewed as a collection of column vectors, we use $D_{\V}$ to denote $D_{\v_1,\dots,\v_k}$, i.e.,
\[ D_{\V}=\prod_{j=1}^k\sum_{i=1}^n V_{ij}\partial_i. \]
We denote the gradient of a function or polynomial $g$ by $\nabla g$ and the Hessian of $g$ by $\Hess{g}$. For a polynomial $g\in \R[z_1, \hdots, z_n]$ and a number $c\in \R$ we write $\eval{g(z_1,\dots,z_n)}_{z_1=c}$ to denote the restricted polynomial in $z_2,\dots,z_n$ obtained by setting $z_1=c$. For a polynomial $g(z_1,\dots,z_n)$, we define $\supp(g)\subset \Z_{\geq 0}^n$ as the set of vectors $\bmkappa=(\kappa_1,\dots,\kappa_n)\in \Z_{\geq 0}^n$ such that the coefficient of the monomial $\prod_i z_i^{\kappa_i}$ is nonzero. The convex hull of $\supp(g)$ is known as the Newton polytope of $g$. We call a polynomial multiaffine if the degree of each variable is at most one.

\subsection{Log-Concave Polynomials}

We say that a polynomial $g\in \R[z_1, \dots, z_n]$ with nonnegative coefficients is \emph{log-concave} if $\log(g)$ is a concave function over $\R_{>0}^n$. For simplicity, we consider the zero polynomial to be log-concave. 
Equivalently, $g$ is log-concave if for any two points $\v, \w\in \R_{\geq 0}^n$ and $\lambda\in [0,1]$ we have
\[ g(\lambda \v+(1-\lambda)\w)\geq g(\v)^\lambda \cdot g(\w)^{1-\lambda}. \]
A polynomial $g$  with nonnegative coefficients is log-concave if and only if the Hessian of $\log(g)$ is negative semidefinite at all points of $\R_{\geq 0}^n$ where it is defined. In particular, the set of log-concave polynomials is closed in the Euclidean space of polynomials of degree $\leq d$. To see this, note that
the nonnegativity of the coefficients of $g$ ensure that if $g\neq 0$, then $\log(g)$ is defined at all points of $\R_{> 0}^n$. 
Furthermore the entries of the Hessian of $\log(g)$ at a point in $\R_{>0}^n$ are continuous functions in the coefficients of $g$. The closed-ness of the set of log-concave polynomials then follows from the closed-ness of the cone of negative semidefinite matrices. The basic operation that preserves log-concavity is an affine transformation of the inputs.
\begin{lemma}
	\label{prop:affinetransform}
Let $g\in \R[z_1,\dots, z_n]$ be a log-concave polynomial with nonnegative coefficients. 
Then for any affine transformation $T:\R^m\to \R^n$ with $T(\R_{\geq 0}^m)\subseteq \R_{\geq 0}^n$, 
$g(T(y_1,\dots,y_m)) \in \R[y_1,\dots,y_m]$  has nonnegative coefficients and is log-concave.
\end{lemma}
\begin{proof}
We can write $T$ as  $T(\y)=\A\y+\b$ for some $\A\in \R^{n\times m}$ and $\b\in \R^n$, and one can check check that $T(\R_{\geq 0}^m)\subseteq \R_{\geq 0}^n$ if and only if $\A$ and $\b$ have nonnegative entries. It follows that $g(T(y_1,\dots,y_m))$ has nonnegative coefficients.
	
	To check log-concavity, note that $T$ being affine implies that for two points $\v, \w \in \R_{\geq 0}^m$ and $\lambda\in [0,1]$, we have $T(\lambda \v+(1-\lambda)\w)=\lambda T(\v)+(1-\lambda)T(\w)$. It follows that
	\begin{align*} 
		g(T(\lambda \v+(1-\lambda) \w))&=g(\lambda T(\v)+(1-\lambda)T(\w))\\
		&\geq g(T(\v))^\lambda\cdot g(T(\w))^{1-\lambda},
	\end{align*}
	where the inequality follows from log-concavity of $g$.
\end{proof}

\begin{proposition}\label{prop:preservers}
	The following operations preserve log-concavity:
	\begin{parts}
		\item \label{pres:permutation} Permutation: $g\mapsto g(z_{\pi(1)}, \dots, z_{\pi(n)})$ for $\pi\in S_n$.
		\item \label{pres:specialization} Specialization: $g\mapsto g(a,z_2, \dots, z_n)$, where $a\in \R_{\geq 0}$.
		\item \label{pres:scaling} Scaling: $g\mapsto  c\cdot g(\lambda_1 z_1, \dots,\lambda_n z_n)$, where $c, \lambda_1, \dots, \lambda_n\in \R_{\geq 0}$.
		\item \label{pres:expansion} Expansion: $g(z_1,\dots,z_n) \mapsto g(y_1+y_2+\dots+y_{m}, z_2,\dots,z_n) \in \R[y_1, \dots, y_m, z_2, \dots, z_n]$.
		\item \label{pres:multiplication} Multiplication: $g, h\mapsto g\cdot h$ where $g, h$ are log-concave.
	\end{parts}
\end{proposition}

\begin{proof} \Cref{pres:permutation,pres:specialization,pres:scaling,pres:expansion}  follow by choosing an appropriate affine transformation $T$ and applying \cref{prop:affinetransform}. For \cref{pres:scaling} we also need to use the elementary fact that scaling by $c\geq 0$ preserves log-concavity.
For \cref{pres:multiplication}, note that $\log (g\cdot h) = \log g + \log h$. Since the sum of any two concave functions is concave, $g\cdot h$ is log-concave.
\end{proof}

In general log-concavity is not preserved under taking derivatives.  For example, $g(z) = z^4/4 + z$ is log-concave on $\R_{>0}$, but $h = \partial g/\partial z = z^3+1$ is not: 
\[
	\frac{\partial ^2 \log(g)}{\partial z}=\frac{-4(z^3-2)^2}{z^2(z^3+4)^2}\leq0 
	\quad\text{ and }\quad
\eval{\frac{\partial^2 \log(h)}{\partial z}}_{z=1} =  \eval{\frac{-3z (z^3-2)}{(z^3+1)^2}}_{z=1}=3/4.
\]
In \cref{sec:logconcavity}, we remedy this by considering completely log-concave polynomials, for which log-concavity is preserved under differentiation. 

\subsection{Matroids}
\label{prelim:matroid}

Let $M=(E,\I)$ be a matroid, as defined in \cref{sec:introduction}. For any set $S\subseteq E$, the rank of $S$, denoted $\rank(S)$, is the size of the largest subset $A\subseteq S$ such that $A\in \I$. The rank of the matroid is the rank of the set $E$, and a set $B\subseteq E$ is a basis of $M$ if and only if $B\in \I$ and $\rank(B) = \rank(E)$.

We say a matroid $M$ is \emph{simple} if it has no loops and no parallel elements, meaning that for all pairs $i\neq j\in E$, $\rank(\set{i, j})=2$. The \emph{dual} matroid of $M$ is the matroid  $M^*=(E,\I^*)$ on the same set of elements $E$ whose bases are the complements $E\setminus B$ of bases $B$ of $M$ (and whose independent are subsets of those bases). Given two matroids, $M_1=(E_1,\I_1)$ and $M_2=(E_2,\I_2)$, we can also define their \emph{direct sum} to be the matroid
\[
	M_1 \oplus M_2 = \tuple*{E_1\sqcup E_2, \set*{ I_1 \sqcup I_2 \given I_1\in \I_1, I_2 \in \I_2}},
\]
where $\sqcup$ denotes disjoint union. When $M_1, M_2$ have disjoint ground sets, this notion coincides with that of matroid union.

The matroid base polytope $\P_M \subset \R^{E}$ of a matroid $M$ is the convex hull of the indicator vectors of its bases. If $M$ has rank $r$, it can also be defined by the following system of inequalities: 
\begin{equation}
	\P_M=\conv\set{\1_B\given B\in \B_M}=\set*{\p\in \R^E\given
	\begin{array}{lr}
 		\dotprod{\1_{E}, \p}=\sum_{i\in E} p_i = r, \\
 		\dotprod{\1_S, \p}=\sum_{i\in S} p_i \leq \rank(S) & \forall S\subseteq E, \\
 		\dotprod{\1_i, \p}=p_i\geq 0 & \forall i\in E.
	\end{array}}.
	\label{eq:matroidpolytope}	
\end{equation}

While the above description requires exponentially many constraints (one for each set $S\subseteq E$), because of the matroidal structure, one can optimize a linear function over the vertices of $\P_M$ in polynomial time, assuming access to an independence oracle, and thereby construct a separation oracle for the polytope \cite{Cun84}. Using this, one can minimize any convex function over the matroid base polytope in polynomial time. See \textcite{BV04} for background on convex optimization.


For $\bmlambda=(\lambda_1,\dots,\lambda_n) \in \R_{\geq 0}^n$, we define $\bmlambda$-weight of a subset of $E$ 
by taking the product of weights of elements in the subset.  That is, for a set $S\subseteq E$ we use
\begin{equation}\label{eq:weightltM}
	\bmlambda^S=\prod_{i\in S} \lambda_i
\end{equation}
to denote $\bmlambda$-weight of the set $S$. Note that if $g_M$ is the basis generating polynomial of a matroid $M$, as in \cref{eq:gM}, the $\bmlambda$-weight of a basis is the coefficient of the corresponding monomial in $g_M$ after scaling the variables by $\lambda_1,\dots,\lambda_n$, i.e. 
\[
	g_M(\lambda_1 z_1,\dots,\lambda_n z_n)=\sum_{B\in \B_M} \bmlambda^B \prod_{i\in B} z_i.
\]

\subsection{Linear Algebra}

A symmetric matrix $A\in\R^{n\times n}$ is positive semidefinite (PSD), denoted $A\succeq 0$, if for all $v\in\R^n$,
\[ v^\intercal Av \geq 0,\]
and it is positive definite (PD) if the above inequality is strict for all $v\neq 0$. Similarly, $A$ is negative semidefinite (NSD), denoted $A\preceq 0$, if $v^\intercal Av \leq 0$ for all $v\in\R^n$, and negative definite (ND) if $v^\intercal Av <0$ for $v \neq 0$. Equivalently, a real symmetric matrix is PSD (PD, NSD, ND) if its eigenvalues are nonnegative (positive, nonpositive, negative), respectively. 

Let $A\in\R^{n\times n}$ be a symmetric matrix with eigenvalues $\lambda_1 \geq \dots\geq\lambda_n$ and corresponding orthonormal eigenvectors $v_1,\dots,v_n$. Then, we can write $A$ as
\[ A=\sum_{i=1}^n \lambda_i v_iv_i^\intercal.\]

We say that a sequence of real numbers 
 $\beta_{1}\geq \dots\geq \beta_{n}$ \emph{interlaces} 
 $\alpha_{1}\geq \alpha_{2} \geq \dots\geq \alpha_n$ if 
 \[ \alpha_1\geq \beta_{1}\geq \alpha_{2} \geq \dots \geq \beta_{n-1} \geq \alpha_n \geq \beta_n.\]
The following useful theorem is known as Cauchy's interlacing theorem:
\begin{theorem}[{Cauchy's Interlacing Theorem I  \cite[see][Corollary~4.3.9]{HJ13}}]\label{thm:CauchyInterlacing}
	For a symmetric matrix $A\in\R^{n\times n}$ and vector $v\in \R^n$, the eigenvalues of $A$ interlace the eigenvalues of $A+vv^\intercal$.
\end{theorem}	
The following is an immediate consequence:
\begin{lemma}\label{lem:Cauchy1eigenvalue}
	Let $A\in\R^{n\times n}$ be a symmetric matrix and let $P\in\R^{m\times n}$. If $A$ has at most one positive eigenvalue, then $PAP^\intercal$ has at most one positive eigenvalue.
\end{lemma}
\begin{proof}
	Since $A$ has at most one positive eigenvalue, we can write $A=B+vv^\intercal$ for some vector $v\in \R^n$ and some negative semidefinite matrix $B$. Then $PAP^\intercal = PBP^\intercal + Pvv^\intercal P^\intercal$. First, observe that $PBP^\intercal\preceq 0$, since for $x\in\R^m$, $x^\intercal PB^\intercal P x = (Px)^\intercal B (Px)\leq 0$. Second, let $w = Pv\in \R^m$. Then $Pvv^\intercal P^\intercal=ww^\intercal$ and by \cref{thm:CauchyInterlacing}, the eigenvalues of $PBP^\intercal$ interlace the eigenvalues of $PBP^\intercal+(Pv)(Pv)^\intercal$. Since all eigenvalues of $PBP^\intercal$ are nonpositive, $PAP^\intercal=PBP^\intercal+ww^\intercal$ has at most one positive eigenvalue. 
\end{proof}

Another version of Cauchy's Interlacing Theorem can be stated as follows:

\begin{theorem}[{Cauchy's Interlacing Theorem II, \cite[see][Theorem~4.3.17]{HJ13}}]\label{thm:CauchyInterlacingII}
	Let  $A\in\R^{n\times n}$ be symmetric and $B \in \R^{(n-1)\times (n-1)}$ a principal submatrix of $A$. Then the eigenvalues of $B$ interlace the eigenvalues of $A$. That is, 
	\[ \alpha_1\geq \beta_{1}\geq \alpha_{2} \geq \dots \geq \beta_{n-1} \geq \alpha_n, \]
	where  $\alpha_{1},\dots, \alpha_n$ and $\beta_{1},\dots,\beta_{n-1}$ are the eigenvalues of $A$ and $B$, respectively. 
\end{theorem}	

This has the immediate corollary:

\begin{corollary}
	\label{cor:Cauchy2}
	Let $A\in\R^{n\times n}$ be a symmetric matrix. If there is an $(n-1)$-dimensional linear subspace on which the quadratic form $z\mapsto z^\intercal A z$ is nonpositive, then $A$ has at most one positive eigenvalue.
\end{corollary}
\begin{proof}
	After a change of basis, we can assume that the $(n-1)$-dimensional linear subspace on which $A$ is negative semidefinite is spanned by coordinate vectors $e_1, \dots, e_{n-1}$. Then the top-left principal minor of $A$ is negative semidefinite, so by \cref{thm:CauchyInterlacingII}, $A$ has at most one positive eigenvalue. 
\end{proof}

We will also need the following lemma inspired by arguments of \textcite{HW17}. 

\begin{lemma}
	\label{lem:almostNSD}
	Let $A\in \R^{n\times n}$ be a real symmetric matrix with nonnegative entries and at most one positive eigenvalue. Then for every $v\in \R_{\geq 0}^n$, the $n\times n$ matrix 
	\[
		(v^\intercal A v) \cdot A - t (Av)(Av)^\intercal
	\]
	is negative semidefinite for all $t\geq 1$. 
\end{lemma}
\begin{proof}
	Note that since the set of negative semidefinite matrices is a closed convex cone, it suffices to prove this for $v\in \R_{>0}^n$ and $t=1$. If $A$ is the zero matrix, the claim is immediate. Otherwise, $v^\intercal A v >0$. Let $w\in \R^n$ and consider the $2\times n$ matrix $P$ with rows $v^\intercal$ and $w^\intercal$. Then 
	\[
		P A P^\intercal=
		\begin{bmatrix} 
			v^\intercal Av & v^\intercal Aw \\ 
			w^\intercal Av & w^\intercal Aw
		\end{bmatrix}.
	\]
	By \cref{lem:Cauchy1eigenvalue}, $P A P^\intercal$ has at most one positive eigenvalue. On the other hand, the diagonal entry $v^\intercal Av$ is positive, so \cref{thm:CauchyInterlacingII} implies that $P A P^\intercal$ has a positive eigenvalue, meaning that it must have exactly one. It follows that 
	\begin{equation}
		\label{eq:Delta}
		\det(P A P^\intercal)=(v^\intercal A v) \cdot(w^\intercal A w)  - (w^\intercal Av) \cdot (v^\intercal Aw)\leq 0.
	\end{equation}
	Thus $w^\intercal ((v^\intercal A v) \cdot A - (Av)(Av)^\intercal) w \leq 0$ for all $w\in \R^n$. 
\end{proof}

\subsection{Entropy and External Fields}
\label{subsec:MaxEntrop}

For any probability distribution $\mu$ supported on a finite set $\Omega$, its \emph{entropy}, $\H(\mu)$, is defined to be
\[ \H(\mu)=\sum_{\omega\in \Omega} \mu(\omega)\log \frac{1}{\mu(\omega)}. \]
For a number $p\in [0, 1]$, we also use the shorthand
\[\H(p)=p\log \frac{1}{p}+(1-p)\log \frac{1}{1-p}\]
to denote the entropy of the Bernoulli distribution with parameter $p$. See \textcite{Cov06} for background on entropy and its properties. One of the basic facts we will use about entropy is subadditivity.
\begin{fact}
	\label{fact:subadditivity}
	Suppose that $X$ and $Y$ are finitely supported, not necessarily independent, random variables. Let the joint distribution of $(X, Y)$ be $\mu$, and let $\mu_X$ and $\mu_Y$ denote the marginal distributions of $X$ and $Y$ respectively. Then
	\[ \H(\mu)\leq \H(\mu_X)+\H(\mu_Y), \]
with equality if and only if $X$ and $Y$ are independent.
\end{fact}

Let $\mu:2^{[n]}\to\R_{\geq 0}$ be a nonnegative function; we say $\mu$ is a \emph{probability distribution} if $\sum_{S\subseteq [n]} \mu(S) = 1$. The \emph{support} of $\mu$, denoted $\supp(\mu)$, is the collection of $S\subseteq [n]$ for which $\mu(S)\neq 0$, and the \emph{Newton polytope} $\P_\mu \subset \R^n$ of $\mu$ is the convex hull of the indicator vectors in its support, i.e., $\P_\mu=\conv\set{\1_S\given S\in \supp(\mu)}$. The entropy of $\mu$ equals \[ \H(\mu)=\sum_{S\in\supp(\mu)} \mu(S) \log \frac1{\mu(S)}.\]

To use entropy for approximate counting, we use the following fundamental fact: 
\begin{proposition}\label{prop:EntropyCounting}
	If $u:2^{[n]}\to\R_{\geq 0}$ is the uniform distribution over $S \in \supp(u)$, then $\H(u)$ equals the $\log$ of the number of elements in the support of $u$ and this is an upper bound for the entropy of any distribution $\mu$ with $\supp(\mu)\subseteq \supp(u)$. That is, for any distribution $\mu:2^{[n]}\to\R_{\geq 0}$,   
	\[ \H(\mu)\leq \log\parens*{\card{\supp(\mu)}},\]
	with equality when $\mu$ is the uniform distribution over its support. 
\end{proposition}

We define the generating polynomial of $\mu$ to be the multiaffine polynomial 
\[ g_{\mu}(z_1, \dots, z_n)=\sum_{S\subseteq [n]} \mu(S) \prod_{i\in S} z_i. \]
The nonnegative function or probability distribution $\mu$ is said to be \emph{log-concave} if its generating polynomial $g_{\mu}$ is log-concave as a function on the positive orthant. The marginal probability of an element $i$, $\mu_i$, is the probability that $i$ is in a random sample from $\mu$,
\[ \mu_i = \PrX{S\sim\mu}{i\in S} = \sum_{S \ni i} \mu(S) =  \eval{\partial_{z_i}g_\mu(z_1,\dots,z_n)}_{z_1=\dots=z_n=1}.\]

For a collection of positive numbers $\bmlambda=(\lambda_1,\dots,\lambda_n)$, the $\bmlambda$-external field applied to $\mu$ is a 
probability distribution $\ef{\bmlambda}{\mu}:2^{[n]}\rightarrow \R_{\geq 0}$ where for every $S$,
\begin{equation}\label{eq:ExtField}
	\PrX{\ef{\bmlambda}{\mu}}{S}\propto\bmlambda^S\cdot \mu(S)=\parens*{\prod_{i\in S} \lambda_i} \cdot \mu(S).
\end{equation}
As with matroid weights, we note that 
\[g_{\ef{\bmlambda}{\mu}}(z_1, \dots, z_n)\propto \sum_{S\subseteq [n]} \bmlambda^S \mu(S) \prod_{i\in S} z_i=g_{\mu}(\lambda_1 z_1, \dots, \lambda_n z_n). \]

The following theorem has been rediscovered many times:
\begin{theorem}[\cite{AGMOS10,SV14}]\label{thm:maxentropy}
Let $\mu:2^{[n]}\to\R_{\geq 0}$ be a function. For any point $\p$ in the polytope $\P_\mu$ and for any $\epsilon>0$, there exist weights $\lambda_1,\dots,\lambda_n \in \R_{> 0}$ such that the marginal probabilities of $\ef{\bmlambda}{\mu}$ are within $\epsilon$ of $\p$, i.e., for all $i \in [n]$,
\[ \abs*{p_i - \PrX{S\sim \ef{\bmlambda}{\mu}}{i\in S}} \leq \epsilon.\]
\end{theorem}

If $\p$ is in the relative interior of the polytope $\P_{\mu}$, it turns out that we can take $\epsilon=0$ in the above theorem. In either case though, we have the following. 

\begin{corollary}\label{cor:extFieldClosure}
	Let $\mu:2^{[n]}\to\R_{\geq 0}$ be a nonnegative function and let $\p \in \P_\mu$. There is a probability distribution $\tmu$ with marginals $\p$ such that $\supp(\tmu)\subseteq\supp(\mu)$. Moreover $\tmu$ can be obtained as a limit of distributions of the form $\ef{\bmlambda}{\mu}$ for $\bmlambda\in \R_{>0}^n$.
\end{corollary}

\begin{proof}
Note that the set of probability distributions $\tmu:2^{[n]} \to\R_{\geq 0}$ with $\supp(\tmu)\subseteq\supp(\mu)$ is compact. By \cref{thm:maxentropy}, for any $\epsilon>0$, there exist weights $\bmlambda=(\lambda_1,\dots,\lambda_n)\in \R_{>0}^n$ such that for each $i$, the $i$-th marginal of $\ef{\bmlambda}{\mu}$ is within $\epsilon$ of $p_i$. By passing to a convergent subsequence, we can take $\tmu:2^{[n]} \to\R_{\geq 0}$ to be the limit of the distributions $\ef{\bmlambda}{\mu}$ as $\epsilon\to 0$. Then the marginals of $\tmu$ are exactly $\p$, and, since the support of $\ef{\bmlambda}{\mu}$ is contained in the support of $\mu$ for all $\bmlambda$, it follows that the support of $\tmu$ is also contained in the support of $\mu$. 
\end{proof}

The following corollary follows from \cref{prop:preservers,thm:maxentropy}.
\begin{corollary}
\label{cor:maxentropy}
	Let $\mu:2^{[n]} \to\R_{\geq 0}$ be log-concave and let $\p \in \P_\mu$. There is a log-concave probability distribution $\tmu$ with marginals $\p$ such that $\supp(\tmu)\subseteq\supp(\mu)$.  Moreover, $\tmu$ can be obtained as the limit of distributions $\ef{\bmlambda}{\mu}$ for $\bmlambda \in \R_{>0}^n$.
\end{corollary}
\begin{proof}
	For every $\bmlambda \in \R_{>0}^n$, 
	\[ g_{\ef{\bmlambda}{\mu}}(z_1,\dots,z_n)\propto g_{\mu}(\lambda_1 z_1,\dots, \lambda_n z_n).\]
	Therefore by \cref{prop:preservers}, \cref{pres:scaling}, each $\ef{\bmlambda}{\mu}$ is log-concave. Since the set of log-concave polynomials is closed in the Euclidean topology, it follows that the limit $\tmu$ given in the proof of \cref{cor:extFieldClosure} is also log-concave. 
\end{proof}

\begin{remark}
	Although we will not use this fact, it is worth mentioning that the distribution $\tmu$ promised by \cref{cor:maxentropy} can be obtained by solving a maximum entropy program:
	\[ \argmax_{\tmu}\set*{\sum_{S} \tmu(S)\log \frac{\mu(S)}{\tmu(S)} \given \forall i\in [n]:\tmu_i=p_i} \]
\end{remark}

The entropy has the following interesting relationship with geometric programs.

\begin{lemma}[\cite{SV14}]
	\label{lem:geometric-program}
	Let $\mu:2^{[n]}\to \R_{\geq 0}$ have generating polynomial $g_\mu\in \R[z_1,\dots,z_n]$. Let $\p$ be a point in the Newton polytope of $\mu$. Then
	\[ 
		\log\parens*{\inf_{\z\in \R_{>0}^n} \frac{g_\mu(\z)}{\prod_i z_i^{p_i}}}=\sum_{S} \tmu(S)\log \frac{\mu(S)}{\tmu(S)},
	\]
	where $\tmu$ is the probability distribution given by \cref{cor:extFieldClosure,cor:maxentropy}. In particular, if $\mu$ is the indicator function of a family $\B\subseteq 2^{[n]}$, i.e., if $\mu$ only takes values in $\set{0,1}$, then the above quantity is the same as the entropy $\H(\tmu)$. 
\end{lemma}

For us, this will be of particular interest when $\mu$ is the indicator function over its support. In this case, by \cref{prop:EntropyCounting}, the infimum in \cref{lem:geometric-program} will give a lower bound for the entropy of $\mu$ and thus the $\log$ of the size of its support.

	\section{Hodge Theory for Matroids}

\label{sec:hodge}
In this section we review several recent developments on combinatorial Hodge theory by \textcite{AHK15,HW17}. The main result we prove in this section is \cref{thm:HessNSD}. Later in \cref{sec:logconcavity}, we use this to prove that the generating polynomial of the bases of any matroid is a log-concave function over the positive orthant.

In this section, we take all matroids to be simple. To describe the algebraic tools used \cite{AHK15,HW17}, we introduce a little more matroid terminology, namely the theory of flats. A subset $F\subseteq E$ is a flat of $M = (E,\I)$ if it is a maximal set with rank equal to $\rank(F)$, i.e., for any $i\not\in F$, $\rank(F\cup\{i\}) =  \rank(F)+1$.  In particular, $F = E$ is the unique flat of rank equal to $\rank(M)$. 

We say that a flat $F$ is \emph{proper} if $F\neq E$. Flats $F_1,F_2$ are \emph{comparable} if $F_1\subseteq F_2$ or $F_2\subseteq F_1$ and they are \emph{incomparable} otherwise. A \emph{flag} of $M$ is a strictly monotonic sequence of  nonempty proper flats of $M$, $F_1\subsetneq F_2 \subsetneq \dots \subsetneq F_{l}$. Note that any flag of $M$ has at most $\rank(M)-1$ flats.

\subsection{The Chow ring}\label{sec:ChowRing}

Here we go through some of the commutative algebra used by \textcite{AHK15} and explain a special case of their main theorem. Following the set up of \textcite{AHK15}, for a matroid $M$ of rank $r+1$ on the ground set $E$, define the \emph{Chow ring} to be the ring 
\[
	A^*(M)=\R\brackets{x_F\given F \text{ is a nonempty proper flat of } M}
\] whose generators $x_F$ satisfy the relations 
\[
 x_{F_1}x_{F_2} =0  \text{ for all incomparable $F_1,F_2$}\qquad\text{and}\qquad
 \sum_{F \ni i} x_{F} - \sum_{F\ni j} x_{F} = 0 \text{ for all $i,j\in E$}.
\]
Since these relations are homogeneous polynomials in the generators, $A^*(M)$ is a graded ring, and we use $A^d(M)$ to denote homogeneous polynomials in $A^*(M)$ of degree $d$. It turns out that the top degree part, $A^{r}(M)$, is a one-dimensional vector space over $\R$, 
and we write  ``$\deg$'' for the isomorphism $A^{r}(M) \simeq \R$ determined by the property that
$ \deg(x_{F_1}\dots x_{F_{r}})=1$
for any flag $F_1\subsetneq F_2\subsetneq \dots \subsetneq F_{r}$ of $M$.

A function $f$ on the set of nonempty proper subsets of $E$ is said to be \emph{strictly submodular} if 
\[ f(S)+f(T) > f(S\cap T) + f(S\cup T) \]
for any two incomparable subsets $S,T$ of $E$, where we take $f(\emptyset) = f(E) = 0$. We say $f$ is \emph{submodular} if it satisfies the weak form of the above inequality, i.e., with possible equality. We remark that our notion of submodularity differs slightly from the conventional notion, in that we effectively require, in addition to conventional submodularity, that $f$ takes a value of $0$ on $\emptyset, E$.

Define the open convex cone 
\[
	K(M)=\set*{\sum_F f(F) x_{F}  \given f \text{ is strictly submodular}}\subset A^1(M),
\]
where the sum is over all flats of $M$. We will use $\barK(M)$ to denote the Euclidean closure of this cone, namely the elements of $ A^1(M)$ whose coefficients give a submodular function on subsets of $E$. 

The following is a special case of one of the main theorems of \textcite{AHK15}.
\begin{theorem}[{\cite[Theorem~8.9]{AHK15}}]\label{thm:AHK}
	Let $M$ be a simple matroid of rank $r+1$. For any $\ell_0, \ell_1,\dots,\ell_{r-2}$ in $K(M)$, consider the symmetric bilinear form $Q_{\ell_1,\dots,\ell_{r-2}}: A^1(M)\times A^1(M) \rightarrow \R$ defined by 
	\[ Q_{\ell_1,\dots,\ell_{r-2}}(v,w)=\deg(v\cdot \ell_1 \cdot \ell_2 \cdot \dots \ell_{r-2}\cdot w). \]
	Then, as a quadratic form on $A^1(M)$, $Q_{\ell_1,\dots,\ell_{r-2}}$ is negative definite on the kernel of $Q_{\ell_1,\dots,\ell_{r-2}}(\ell_0, \cdot)$, i.e., on 
	\[ \set*{v\in A^1(M)\given Q_{\ell_1,\dots,\ell_{r-2}}(\ell_0, v)=0}. \]
\end{theorem}
Note that the kernel of $Q_{\ell_1,\dots,\ell_{r-2}}(\ell_0, \cdot)$ has codimension one, implying that the operator $Q_{\ell_1,\dots,\ell_{r-2}}$ has at most one nonnegative  eigenvalue.

Also, observe that the above result naturally extends to taking $\ell_0, \ell_1,\dots,\ell_{r-2}$ in the Euclidean closure $\barK(M)$ at the expense of having the slightly weaker guarantee that the operator $Q_{\ell_1,\dots,\ell_{r-2}}$ will be negative \emph{semidefinite} on the kernel of $Q_{\ell_1,\dots,\ell_{r-2}}(\ell_0, \cdot)$. 
In this case, $Q_{\ell_1,\dots,\ell_{r-2}}$ has at most one positive eigenvalue.


\subsection{Graded M\"obius Algebra}
To connect the Chow ring with the basis generating polynomial of a matroid, we introduce another algebra used by \textcite{HW17}. Here we take $M$ to be a simple matroid of rank $r$ on the ground set $[n]$. For flats $F_1, F_2$, define $F_1\vee F_2$ to be the inclusion-minimal flat containing $F_1\cup F_2$.

Let $B^*(M)$ denote the ring $\R\brackets{y_F\given F \text{ is a flat of $M$}}$ whose generators $y_F$ satisfy the relations
\begin{equation}\label{def:mobiusalgebra}
	y_{F_1}\cdot y_{F_2}=
	\begin{cases}
		y_{F_1\vee F_2} & \text{if } \rank(F_1)+\rank(F_2)=\rank(F_1\vee F_2),\\
		0& \text{otherwise},
	\end{cases}
\end{equation}
for all pairs of flats $F_1, F_2$.  These relations imply that for any flat $F$ and any basis $I_F$ of $F$, $y_F$ equals the product  $\prod_{i\in I_F} y_i$, where $y_i=y_{\set{i}}$. It follows that we can instead take $y_1, \dots, y_n$ as generators of $B^*(M)$ and that the relations on these generators are degree-homogeneous. Then $B^*(M)$ is a graded algebra and we use $B^d(M)$ to denote the homogeneous polynomials of degree $d$ in $B^*(M)$.

\Textcite{HW17} relate this to the Chow ring as follows. Let $\barM$ denote the matroid of rank $r+1$ on ground set $\{0,1,\dots, n\}$ obtained by adding $0$ as a coloop. Its independent sets have the form $I$ or $\{0\}\cup I$, where $I$ is independent in $M$. In the Chow ring of $\barM$, for each $i=1, \dots, n$, define the degree one element
\[\beta_i=\sum_{F:i\in F, 0\notin F} x_{F}\in A^1(\barM),\]
where the sum is taken over flats $F$ of $\barM$ for which  $i\in F$ and $0\notin F$. Since the indicator function of the condition $i\in F$ and $0\notin F$ is submodular, $\beta_i$ belongs to the closed convex cone $\overline{K}(\barM)$.

\Citeauthor{HW17} use this to establish the following relationship between $A^*(\barM)$ and $B^*(M)$. 
\begin{theorem}[{\cite[Prop 9]{HW17}}]
	There is a unique injective graded $\R$-algebra homomorphism
	\[ \varphi:B^*(M)\mapsto A^*(\barM)\qquad\text{with}\qquad\varphi(y_i) = \beta_i.\]
\end{theorem}
Note that for any basis $B$ of $M$, $\prod_{i\in B} y_i  = y_{[n]}$ is nonzero in $B^r(M)$. On the other hand, if $S \subseteq [n]$  is a dependent set of the matroid $M$, then $\prod_{i\in S} y_i$ is zero in  $B^*(M)$. Then from the existence and injectivity of this map, it follows that for any set $B \subseteq [n]$ with $\card{B}=r$, up to global scaling by a positive real number,
\[
	\deg\parens*{\prod_{i\in B} \beta_i}=
	\begin{cases}
		1 & \text{if $B$ is a basis of $M$},\\
		0 & \text{otherwise,}
	\end{cases}
\] 
where $\deg:A^r(\barM)\rightarrow \R$ is the isomorphism discussed in \cref{sec:ChowRing}.

This is particularly useful for us because of the following connection with differential operators on the basis generating polynomial $g_M(z_1, \dots, z_n)$.

\begin{proposition}
\label{prop:ChowToDer}
	For a matrix $\V \in \R^{n\times r}$ with columns $\v_1\given \dots\given \v_r \in \R^n$,  
	\[
		\deg\parens*{\prod_{j=1}^r  \sum_{i=1}^n V_{ij} \beta_i}=\partial_{\v_1}\cdots\partial_{\v_r}g_M(\z). 
	\]
	Furthermore, for any $0\leq k\leq r$ and $\bmlambda=(\lambda_1,\dots,\lambda_n)\in \R^n$, 
	\[
		\deg\parens*{\parens*{\sum_{i=1}^n\lambda_i \beta_i}^{r-k} \cdot \prod_{j=1}^k \sum_{i=1}^n V_{ij} \beta_i}=(r-k)! \cdot \eval{\partial_{\v_1}\cdots\partial_{\v_k} g_M(\z)}_{\z = \bmlambda}. 
	\]
\end{proposition}

\begin{proof}
	For the first claim, note that both sides are linear in each $\v_i$, so it suffices to prove the claim when the columns are unit coordinate vectors $\brackets{\v_1\given \dots\given \v_r} = \brackets{\1_{i_1} \given \dots \given \1_{i_r}}$. In this case we see that
	\[
		\deg\parens*{\beta_{i_1} \cdots \beta_{i_r}}=\partial_{i_1}\cdots \partial_{i_r} g_M(\z)=
		\begin{cases}
			1 & \text{if $\set{i_1, \dots, i_r}$ is a basis of $M$},\\
			0 & \text{otherwise.}
		\end{cases}
	\] 
	The general case then follows from linearity in each column. 

	For the second claim, again both sides are linear in each of the first $k$ columns, $\v_1,\dots, \v_k$, so we can consider $\brackets{\v_1\given\dots\given \v_k} = \brackets{\1_{i_1} \given \dots \given \1_{i_k}}$. Let $I = \set{i_1, \dots, i_k}$. If $\rank(I)<k$, then $\prod_{i\in I}\beta_i$ is zero in $A^*(\barM)$ and similarly $(\prod_{i\in I}\partial_i) g_M(\z)$ is zero.  Otherwise, we find that  
	\[
		\deg\parens*{\parens*{\sum_{i=1}^n\lambda_i \beta_i}^{r-k}  \prod_{i\in I } \beta_{i}}=(r-k)!\cdot \sum_{B\supseteq I}\lambda^{B \setminus I}= (r-k)! \cdot \eval{\parens*{\prod_{i\in I} \partial_{i}} g_M(\z)}_{\z = \bmlambda},
	\]
	where the middle sum is taken over bases $B\in \B(M)$ containing the set $I$. 
\end{proof}

We can then translate \cref{thm:AHK} into a statement about $g_M(\z)$.
\begin{theorem}
\label{thm:HessNSD}
	Let $M$ be a simple matroid of rank $r$ on the ground set $[n]$. For any $0\leq k\leq r-2$, matrix of nonnegative real numbers $\V \in \R_{\geq 0}^{n\times k}$, and any $\bmlambda \in \R_{\geq 0}^n$, the symmetric bilinear form $q_{\V, \bmlambda}:\R^n\times \R^n \to \R$ given by 
	\[
		q_{\V, \bmlambda}(\a,\b)= \eval{\partial_{\a}\partial_{\b}D_{\V} g_M(\z)}_{\z = \bmlambda}
	\]
	is negative semidefinite on the kernel of $q_{\V, \bmlambda}(\bmlambda,\cdot)$.  In particular, the Hessian of $D_{\V}g_M(\z)$ evaluated at $\z=\bmlambda$ has at most one positive eigenvalue. 
\end{theorem}

\begin{proof}
	For $1\leq j \leq k$, define $\ell_j = \sum_{i=1}^n V_{ij}\beta_{i}$ and for $k<j \leq r-2$, define  $\ell_j = \sum_{i=1}^n\lambda_i \beta_i$. For each $i$, $\beta_i$ belongs to the convex cone $\barK(\barM)$, so by the nonnegativity of $V_{ij}$ and $\lambda_i$, so does each $\ell_j$. By \cref{prop:ChowToDer}, $q_{\V, \bmlambda}$ equals the restriction of the symmetric bilinear form $Q_{\ell_1, \dots, \ell_{r-2}}$ to the subspace of $A^1(\barM)$ spanned by $\set{\beta_1, \dots, \beta_n}$. That is, for all $\a, \b \in \R^n$, 
	\[
		q_{\V, \bmlambda}(\a,\b)=\frac{1}{(r-k-2)!}Q_{\ell_1, \dots, \ell_{r-2}}\parens*{\sum_i a_i \beta_i,\sum_i b_i \beta_i}. 
	\]
	Let $\ell_0 = \sum_i \lambda_i \beta_i$. By \cref{thm:AHK}, $Q_{\ell_1, \dots, \ell_{r-2}}$ is negative semidefinite on the kernel of  $Q_{\ell_1, \dots, \ell_{r-2}}(\ell_0, \cdot)$, implying that $q_{\V, \bmlambda} $ is negative semidefinite on the kernel of $q_{\V, \bmlambda}(\bmlambda,\cdot)$.

	Finally, note that the Hessian of $D_\V g_M(\z)$ evaluated at $\z=\bmlambda$ is the $n\times n$ matrix representing the bilinear form $q_{\V, \bmlambda}$ with respect to the coordinate basis. Since it is negative semidefinite on a linear subspace of dimension $n-1$, namely the kernel of $q_{\V, \bmlambda}(\bmlambda, \cdot)$, \cref{cor:Cauchy2} implies that it has at most one positive eigenvalue. 
\end{proof}

In the next section we use the above statement to show that generating polynomials of matroids are log-concave and remain log-concave under directional derivatives along directions in the positive orthant.

\section{Completely Log-Concave Polynomials}
\label{sec:logconcavity}

We call a polynomial $g\in \R[z_1, \dots, z_n]$ \emph{completely log-concave} if for every $k\geq 0$ and nonnegative matrix $\V\in \R_{\geq 0}^{n\times k}$, $D_\V g(\z)$ is nonnegative and log-concave as a function over $\R_{>0}^n$, where
\[D_\V g(\z)=\parens*{\prod_{j=1}^k\sum_{i=1}^n V_{ij} \partial_{i}} g(\z).\]
Note that for $k=0$, this condition implies log-concavity of $g$ itself. We call a distribution $\mu:2^{[n]}\to \R_{\geq 0}$ completely log-concave if and only if $g_\mu$ is completely log-concave.

\begin{remark}
	Related notions of ``strongly log-concave'' and ``Alexandrov-Fenchel'' polynomials were studied in the work of \textcite{Gur08, Gur09} to design approximation algorithms for mixed volume of polytopes and to show Newton-like inequalities for coefficients of these polynomials. \Textcite{Gur08} also mentions that because a positive combination of convex polytopes is a convex polytope, the stronger property we call complete log-concavity is satisfied for the volume polynomial. Unlike strong log-concavity, complete log-concavity is readily seen to be preserved under many useful operations. Completely log-concave polynomials and their properties will be the subject of a future work, but in this section, we will study the main properties we need for the analysis of our counting algorithm.
\end{remark}

Note that complete log-concavity implies nonnegativity of the coefficients of $g$. This is because the coefficient of $\prod_{i}z_i^{\kappa_i}$ in $g$ is a positive multiple of $\eval{\partial_{1}^{\kappa_1}\cdots \partial_{n}^{\kappa_n}g(\z)}_{\z=0}$.

To verify complete log-concavity for a polynomial with nonnegative coefficients, we only have to check log-concavity of order $k$ derivatives for $k\leq r-2$. For $k\geq r$, $D_\V g(\z)$ is a nonnegative constant and for $k=r-1$, it is a linear function with nonnegative coefficients in $z_1,\dots, z_n$. 

The main result of this section is that the basis generating polynomial of any matroid is completely log-concave. 
\begin{theorem}\label{thm:gMlogconcave}
	For any matroid $M$, $g_M(\z)$ is completely log-concave over the positive orthant.
\end{theorem}

First, we show that the statement holds when $M$ is a simple matroid. To do this, we use a corollary of Euler's identity, which states that if a polynomial $g(\z)$ is homogeneous of degree $d$ then 
\begin{equation}
\label{eq:EulerId}
	\dotprod{\nabla g, \z}=\sum_{i=1}^n z_i \partial_i g=d \cdot g(\z).
\end{equation}

\begin{corollary}[Euler's identity]\label{cor:Euler}
	If $g\in \R[z_1,\dots, z_n]$ is homogeneous of degree $d$, then 
	\[
		\Hess{g}\cdot \z=(d-1)\cdot \nabla g \qquad\text{and}\qquad \z^\intercal\cdot \Hess{g}\cdot \z=d(d-1)\cdot g.
	\]
\end{corollary}
\begin{proof}
	The $i$-th entry of the vector $\Hess{g}\cdot \z$ equals $\sum_{j=1}^n z_j \cdot \partial_i \partial_j g$.  Since $\partial_i g$ is homogeneous of degree $d-1$, it follows by Euler's identity, \cref{eq:EulerId}, that this equals $(d-1)\partial_i g$. Multiplying by $\z^\intercal$ and using \cref{eq:EulerId} again gives the second claim. 
\end{proof}

\begin{lemma}
\label{lem:freematlogconcave}
	If $M$ is a simple matroid, then $g_M(\z)$ is completely log-concave.
\end{lemma}
\begin{proof}
	Suppose that $M$ is a simple matroid of rank $r$ on ground set $[n]$. Let $0\leq k \leq r-2$ and take a nonnegative matrix $\V \in \R_{\geq 0}^{n\times k}$. We need to show that for any $\bmlambda\in \R_{>0}^n$, the Hessian of $\log(D_\V g_M(\z))$ is negative semidefinite at the point $\z =\bmlambda$. Note that for $h\in \R[z_1, \dots, z_n]$
	\[
		h^2 \cdot \Hess{\log(h)}=\begin{bmatrix} h \cdot \partial_{i} \partial_{j} h-\partial_{i}h\cdot \partial_{j} h\end{bmatrix}_{1\leq i,j \leq n}=h\cdot \Hess{h}-(\nabla h)(\nabla h)^\intercal.
	\]
	Now let $h = D_\V g_M(\z)$ and consider the quadratic form $q_{\V,\bmlambda} (\a, \b) = \partial_{\a} \partial_{\b} h(\bmlambda)$ as in  \cref{thm:HessNSD}. This is represented by the Hessian matrix of $h$ at $\z = \bmlambda$: 
	\[
		\eval{\Hess{h}}_{\z = \bmlambda}= \begin{bmatrix}\partial_{i} \partial_{j} h(\bmlambda) \end{bmatrix}_{1\leq i,j \leq n}.
	\]
	By \cref{thm:HessNSD} this matrix has at most one positive eigenvalue. Since it also has nonnegative entries, we can apply \cref{lem:almostNSD} with  $A = \eval{\Hess{h}}_{\z = \bmlambda}$ and $v =  \bmlambda$. Since $h(\z)$ is homogeneous of degree $r-k$, \cref{cor:Euler} implies that 
	\[
		v^\intercal A v = (r-k)(r-k-1)h(\bmlambda)\quad\text{and}\quad(Av)(Av)^\intercal=(r-k-1)^2\cdot (\nabla h(\bmlambda))(\nabla h(\bmlambda))^\intercal.
	\]
	Then \cref{lem:almostNSD} states that the matrix $(v^\intercal A v) \cdot A-t(Av)(Av)^\intercal=$
	\[
		(r-k)(r-k-1)\cdot \eval{\parens*{h\cdot \Hess{h}-t\cdot \frac{r-k-1}{r-k}\cdot (\nabla h)(\nabla h)^\intercal}}_{\z=\bmlambda}
	\]
	is negative semidefinite for all $t\geq 1$. Taking $t = \frac{r-k}{r-k-1}$ then shows that $h(\bmlambda)^2 \cdot \eval{\Hess{\log(h)}}_{\z = \bmlambda}$ is negative semidefinite. Thus $h = D_{\V}g_M$ is log-concave on $\R_{>0}^n$.  
\end{proof}

Next we show that similar to \cref{prop:affinetransform}, affine transformations preserve complete log-concavity.

\begin{lemma}
\label{lem:CLCaffine}
	Let $T:\R^m\to \R^n$ be an affine transformation such that $T(\R_{\geq 0}^m)\subseteq \R_{\geq 0}^n$, and let $g\in \R[z_1,\dots, z_n]$ be a completely log-concave polynomial. Then $g(T(y_1,\dots,y_m))\in \R[y_1,\dots,y_m]$ is completely log-concave.
\end{lemma}
\begin{proof}
	As in the proof of \cref{prop:affinetransform}, we must have $T(\y)=\A\y+\b$ where $\A\in \R_{\geq 0}^{n\times m}$ and $\b\in\R_{\geq 0}^n$. It follows that $g(T(\y))$ has nonnegative coefficients. Therefore for any nonnegative matrix $\V$, $D_\V g(T(\y))$ has nonnegative coefficients and is  nonnegative over $\R_{\geq 0}^m$, so we just need to check that it is log-concave.
	
The Jacobian of $T$ at every point is given by $\A$. One can then check that for any $\v\in \R^m$,
	\[ \partial_{\v} g(T(\y))=\eval{(\partial_{\A \v}g(\z))}_{\z=T(\y)}. \]
	Repeated applications of the chain rule yield for any $\v_1,\dots, \v_k\in \R^m$
	\[ \partial_{\v_1}\cdots \partial_{\v_k} g(T(\y))=\eval{(\partial_{\A \v_1}\cdots\partial_{\A \v_k}g(\z))}_{\z=T(\y)}. \]
	So for any $k\geq 0$ and nonnegative matrix of directions $\V\in \R_{\geq 0}^{m\times k}$, we have 
	\[ D_{\V}g(T(\y))= \eval{(D_{\A\V} g(\z))}_{\z=T(\y)}. \]
	Since $\A, \V$ have nonnegative entries, so does $\A\V$. From complete log-concavity of $g$ it follows that $D_{\A\V}g(\z)$ is log-concave over $\R_{>0}^n$. Now \cref{prop:affinetransform} implies that the composition with $T$ remains log-concave.
\end{proof}

\begin{lemma}
\label{lem:CLCpreservers}
	The following operations on polynomials preserve complete log-concavity:
	\begin{parts}
		\item \label{CLCpres:permutation} Permutation: $g\mapsto g(z_{\pi(1)}, \dots, z_{\pi(n)})$ for $\pi\in S_n$.
		\item \label{CLCpres:specialization} Specialization: $g\mapsto g(a,z_2, \dots, z_n)$, where $a\in \R_{\geq 0}$.
		\item \label{CLCpres:scaling} Scaling $g\mapsto  c\cdot f(\lambda_1 z_1, \dots,\lambda_n z_n)$, where $c, \lambda_1, \dots, \lambda_n\in \R_{\geq 0}$.
		\item \label{CLCpres:expansion} Expansion: $g(z_1,\dots,z_n) \mapsto g(y_1+y_2+\dots+y_{m}, z_2,\dots,z_n) \in \R[y_1, \dots, y_m, z_2, \dots, z_n]$.
		\item \label{CLCpres:differentiation} Differentiation: $g\mapsto \partial_{\v} g=\sum_{i=1}^n v_i \partial_i g$ for $\v\in \R_{\geq 0}^n$.
	\end{parts}
\end{lemma}
\begin{proof}
	The proof for \cref{CLCpres:permutation,CLCpres:specialization,CLCpres:scaling,CLCpres:expansion} follows by choosing an appropriate affine transformation $T$ for each part. For \cref{CLCpres:scaling}, we also need to use the elementary fact that scaling by $c\geq 0$ preserves complete log-concavity, since scalar multiplication commutes with differential operator $D_{\V}$.
	
	\Cref{CLCpres:differentiation} follows directly from the definition of complete log-concavity.
\end{proof}

Now, we are ready to prove \cref{thm:gMlogconcave}.
\begin{proof}[Proof of \cref{thm:gMlogconcave}]
	Let $M$ be a matroid of rank $r$ on ground set $[n]$. If $M$ is simple, then the result follows from \cref{lem:freematlogconcave}. 

	Otherwise let $\tilde{M} = (\tilde{E}, \tilde{\I})$ be the simple matroid obtained by deleting loops and identifying each set of parallel elements of $M$. Say each non-loop $i\in [n]$ gets mapped to the element $\psi(i)\in\tilde{E}$. Consider the generating polynomial $g_{\tilde{M}}(\z)\in \R\brackets{z_{e} \given e\in \tilde{E}}$. Each basis of $M$ uses at most one of a set of parallel elements, meaning that the basis generating polynomial of $M$ is obtained from that of $\tilde{M}$ by substituting $z_e \mapsto \sum_{i\in \psi^{-1}(e)} y_i$. That is, if $T:\R^n\to \R^{\tilde{E}}$ is the linear map defined by
	\[ T(\1_i)=\begin{cases}
		0 & \text{if $i$ is a loop},\\
		\1_{\psi(i)} & \text{otherwise},
	\end{cases} \]
	then
	\[
		g_M(y_1, \dots, y_n)=g_{\tilde{M}}(T(y_1,\dots,y_n)).
	\]
	By \cref{lem:CLCaffine}, it follows that $g_M$ is completely log-concave.
\end{proof}

	\section{Entropy of Log-Concave Distributions}
\label{sec:entropylogconcave}

Let $\mu:2^{[n]}\to \mathbb{R}_{\geq 0}$ be a probability distribution on the subsets of the set $[n]$. In other words, $\forall S\subseteq [n]: \mu(S)\geq 0$ and,
\[ \sum_{S\subseteq [n]} \mu(S) = 1.\]
As in \cref{subsec:MaxEntrop}, we consider the multiaffine \emph{generating polynomial} of $\mu$, 
\[ g_\mu(\z) = \sum_{S\subseteq [n]} \mu(S)\cdot \prod_{i\in S} z_i.\]
We say $\mu$ is  \emph{$d$-homogeneous} if $g_\mu$ is a homogeneous polynomial of degree $d$, i.e., $g_\mu(\alpha\z)=\alpha^d g_\mu(\z)$ for any $\alpha\in\R$ and that $\mu$ is \emph{log-concave} and \emph{completely log-concave} if the generating polynomial $g_\mu(\z)$ is log-concave and completely log-concave, respectively. In this section we prove a bound on the entropy of log-concave probability distributions.

Recall that the marginal probability of an element $i$, $\mu_i$ is the probability that $i$ is in a random sample from $\mu$,
\[ \mu_i = \PrX{S\sim\mu}{i\in S} = \eval{\partial_{z_i} g_\mu(z_1,\dots,z_n)}_{z_1=\dots =z_n=1}.\]

Given marginal probabilities $\mu_1,\dots,\mu_n$, it is easy to derive an upper bound on the entropy of $\mu$ by using the subadditivity of entropy, \cref{fact:subadditivity}.

\begin{fact}
	\label{fact:entropyupperbound}
	For any probability distribution $\mu:2^{[n]}\to \R_{\geq 0}$ with marginals $\mu_1,\dots,\mu_n$, we have
	\[
		\H(\mu) \leq \sum_{i=1}^n \parens*{\mu_i \log\frac1{\mu_i} + (1-\mu_i)\log\frac1{1-\mu_i}}=\sum_{i=1}^n \H(\mu_i).
	\]
\end{fact}
The above inequality is tight if the marginals of $\mu$ are independent, i.e., if $\mu$ is a product distribution formed by $n$ independent Bernoulli random variables with parameters $\mu_1,\dots,\mu_n$, i.e., for all sets $S\subseteq [n]$, $\mu(S) = \prod_{i\in S} \mu_i \prod_{i\notin S} (1-\mu_i)$. The main result of this section is a lower bound on the entropy of log-concave distributions, which will imply that the inequality in \cref{fact:entropyupperbound} is tight within a factor of 2 under certain further restrictions.
\begin{theorem}\label{thm:mainentropy}
	For any log-concave probability distribution $\mu:2^{[n]}\to\R_{\geq 0}$ with marginal probabilities $\mu_1,\dots,\mu_n\geq0$, we have
	\[
		\H(\mu)\geq \sum_{i=1}^n \mu_i \log \frac{1}{\mu_i}. 
	\]
\end{theorem}

\begin{example}
Consider the uniform distribution $\mu$ over spanning trees of the complete graph $K_n$. Cayley's formula states that $K_n$ has $n^{n-2}$ spanning trees, so $\H(\mu)=(n-2)\log n$. On the other hand, the marginal probability of every edge in $\mu$ is $\frac2{n}$. By \cref{thm:gMlogconcave}, the generating polynomial $g_{\mu}$ is log-concave. Then the above theorem gives that 
\[  \H(\mu)\geq \sum_{e\in E(K_n)} \frac{2}{n} \log \frac{n}{2} = (n-1)\log \frac{n}{2} .\]
\end{example}
To prove, \cref{thm:mainentropy}, we use Jensen's inequality in order to exploit log-concavity.
	\begin{lemma}[Jensen's Inequality]\label{lem:Jensen}
		Suppose that $f:\R_{\geq 0}^n\to \R\cup\set{-\infty}$ is a concave function, and $X$ is an $(\R_{\geq 0}^n)$-valued random variable with finite support. Then
		\[ f(\Ex{X})\geq \Ex{f(X)}. \]
	\end{lemma}

\begin{proof}[Proof of \cref{thm:mainentropy}]
	In order to apply \cref{lem:Jensen}, we have to specify the concave function $f$ and the random variable $X$. We let $X$ be $\1_S$, where $S$ is chosen randomly according to the distribution $\mu$. In other words, for every $S$, we let
	\[ \Pr{X=\1_S}=\mu(S). \]
	For the function $f$, we will use
	\[ f(z_1,\dots,z_n)=\log g_\mu\parens*{\frac{z_1}{\mu_1},\dots,\frac{z_n}{\mu_n}}. \]
	Note that even though $\mu_i$ could be zero for some $i$, the above expression is still well-defined, since if $\mu_i=0$, then $g_\mu$ does not depend on $z_i$. By \cref{prop:preservers}, \cref{pres:scaling}, the function $f$ is concave over the positive orthant.
	
	First, note that
	\[ \Ex{X}=\ExX{S\sim \mu}{\1_S}=(\mu_1,\dots,\mu_n), \]
	so the left hand side of \cref{lem:Jensen} is
	\begin{equation}
	\label{eq:lhsJensen} 
		f(\Ex{X})=\log g_\mu\parens*{\frac{\mu_1}{\mu_1},\dots,\frac{\mu_n}{\mu_n}}=\log g_\mu(1,\dots,1)=0.
	\end{equation}
	For the right hand side, note that for any $S\in \supp(\mu)$, by the definition of $f$ and $g_\mu$,
	\[ f(\1_S)=\log\parens*{\sum_{T\subseteq S}\mu(T)\prod_{i\in T}\frac{1}{\mu_i}}\geq \log\parens*{\mu(S)\prod_{i\in S}\frac{1}{\mu_i}}=\log\mu(S)+\sum_{i\in S}\log \frac{1}{\mu_i}, \]
	where the inequality follows from monotonicity of $\log$. Now we have
	\begin{align*}
		\Ex{f(X)}&=\sum_{S}\mu(S)f(\1_S)\geq \sum_{S}\mu(S)\log \mu(S)+\sum_S \mu(S) \sum_{i\in S}\log \frac{1}{\mu_i}\\
		&=-\H(\mu)+\sum_{i=1}^n \parens*{\sum_{S\ni i} \mu(S)} \cdot \log \frac{1}{\mu_i}=-\H(\mu)+\sum_{i=1}^n \mu_i \log \frac{1}{\mu_i}.
	\end{align*}
	By \cref{lem:Jensen,eq:lhsJensen}, the above quantity is $\leq 0$. Rearranging yields the desired inequality.
\end{proof}

Next, we discuss several corollaries of \cref{thm:mainentropy}.
\begin{corollary}
\label{cor:homadditiveentropy}
	If $\mu$ is $r$-homogeneous and log-concave, then $\sum_{i=1}^n \H(\mu_i)$ gives an additive $r$-approximation of $\H(\mu)$, i.e.,
	\[
		\sum_{i=1}^n \H(\mu_i)-r \leq \H(\mu)\leq \sum_{i=1}^n \H(\mu_i).
	\]
\end{corollary}
\begin{proof}
	The second inequality is simply \cref{fact:entropyupperbound}. For the first inequality, we use the fact that $(1-p)\log \frac{1}{1-p} \leq p$ for all $0\leq p\leq 1$, which means that
	\begin{equation}
		\label{eq:cardInequality}
		\sum_{i=1}^n  (1-\mu_i)\log\frac{1}{1-\mu_i}\leq \sum_{i=1}^n \mu_i = \ExX{S\sim \mu}{\card{S}},
	\end{equation}
	which for an $r$-homogeneous distribution is equal to $r$.
	Now by applying \cref{thm:mainentropy}, we get
	\[
		\H(\mu)\geq \sum_{i=1}^n \mu_i\log\frac{1}{\mu_i}=\sum_{i=1}^n \H(\mu_i)-\sum_{i=1}^n (1-\mu_i)\log\frac{1}{1-\mu_i}\geq \sum_{i=1}^n\H(\mu_i)-r,
	\]
	as desired.
\end{proof}

For a probability distribution $\mu:2^{[n]}\to\R_{\geq 0}$, define its \emph{dual}, $\mu^*:2^{[n]}\to\R_{\geq 0}$, to be the probability distribution for which the probability of a set is equal to the probability of its complement under $\mu$, i.e.~$\mu^*(S)=\mu([n]\setminus S)$  for all $S\subseteq[n]$. Then for $1\leq i \leq n$, the $i$-th marginal of $\mu^*$ is $\mu^*_i = 1-\mu_i$. 

\begin{corollary}\label{cor:entropyfactor2}
	If $\mu,\mu^*$ are both log-concave probability distributions then $\sum_{i=1}^n \H(\mu_i)$ gives a multiplicative 2-approximation to $\H(\mu)$, i.e.,
	\[ \frac{1}{2}\sum_{i=1}^n \H(\mu_i)\leq \H(\mu)\leq \sum_{i=1}^n \H(\mu_i). \]
\end{corollary}

\begin{proof}	
	Applying \cref{thm:mainentropy} to $\mu$ and $\mu^*$ gives
\[
		\H(\mu)\geq \sum_{i=1}^n \mu_i \log \frac{1}{\mu_i}
		 \ \ \text{ and } \ \ 
		\H(\mu^*)\geq \sum_{i=1}^n (1-\mu_i)\log\frac{1}{1-\mu_i}.
\]
	Since $\H(\mu)=\H(\mu^*)$, averaging the above inequalities gives
	\[ \H(\mu)\geq \frac{1}{2}\sum_{i=1}^n \mu_i \log\frac{1}{\mu_i} +\frac{1}{2} \sum_{i=1}^n (1-\mu_i)\log\frac{1}{1-\mu_i}=\frac{1}{2}\sum_{i=1}^n\H(\mu_i),\]
	as desired. The other inequality follows from \cref{fact:entropyupperbound}.
\end{proof}

Let $\mu$ be the uniform distribution over the bases of a matroid $M$. It follows from \cref{thm:gMlogconcave} that $\mu$ is a log-concave distribution. Furthermore, the dual probability distribution $\mu^*$ is the uniform distribution over the bases of the dual matroid $M^*$, meaning that it is also log-concave. Then  \cref{cor:homadditiveentropy} and \cref{cor:entropyfactor2} immediately yield the following.
\begin{corollary}\label{cor:jensonmatroid}
	Let $M$ be an arbitrary matroid of rank $r$ on ground set $[n]$ and let $\mu$ be the uniform distribution over the bases of $M$. Then $\sum_{i=1}^n\H(\mu_i)$ is both an additive $r$-approximation and multiplicative 2-approximation to $\H(\mu) = \log(|\B_M|)$:
	\[ \max\set*{\frac{1}{2}\sum_{i=1}^n \H(\mu_i), \sum_{i=1}^n \H(\mu_i)-r}\leq \H(\mu)\leq \sum_{i=1}^n \H(\mu_i). \]
\end{corollary}

We will also use the following fact, which enables us to apply \cref{thm:mainentropy} to distributions other than the uniform distribution on $\B_M$.
\begin{lemma}\label{lem:muCLogConv}
	Let $M$ be a matroid on ground set $[n]$ and let $\p$ be a point in $\P_M$. Then there is a distribution $\tmu$ supported on $\B_M$ with marginals $\p$, i.e., $\tmu_i=p_i$, such that both $\tmu$ and $\tmu^*$ are completely log-concave. Furthermore $\tmu$ and $\tmu^*$ can be obtained as the limit of external fields applied to $\mu$ and $\mu^*$, where $\mu$ is the uniform distribution on $\B_M$.
\end{lemma}
\begin{proof}
	If $\mu$ is the uniform distribution over $\B_M$, then $g_{\mu}(\z) = g_M(\z)$, which is completely log-concave by \cref{thm:gMlogconcave}. Similarly, $\mu^*$ is the uniform distribution on the bases of the dual matroid, so $\mu^*$ is also completely log-concave. Furthermore, since $\mu$ and $\mu^*$ are homogeneous distributions, for any $\bmlambda=(\lambda_1,\dots,\lambda_n)\in \R_{>0}^n$, $(\ef{(\lambda_1,\dots,\lambda_n)}{\mu})^*=\ef{(\lambda_1^{-1},\dots,\lambda_n^{-1})}{\mu^*}$, where  $\ef{}{}$ is the external field operation described in \cref{eq:ExtField} of \cref{subsec:MaxEntrop}. 

	By \cref{lem:CLCpreservers}, \cref{CLCpres:scaling}, both $\ef{\bmlambda}{\mu}$ and $(\ef{\bmlambda}{\mu})^*$ are completely log-concave. Now take $\tmu$ to be the distribution promised by \cref{cor:maxentropy} with marginals $\p$. Then $\tmu$ is a limit of distributions $\ef{\bmlambda}{\mu}$, and $\tmu^*$ is the limit of $(\ef{\bmlambda}{\mu})^*$. It follows that both $\tmu$ and $\tmu^*$ are completely log-concave.
\end{proof}

	\section{Max Entropy Convex Programs and Counting Bases of a Matroid}
\label{sec:maxentropymatroidbases}

In this section we prove \cref{thm:mainmatroid}. Let $M$ be a matroid of rank $r$ on ground set $[n]$. Let $\mu:2^{[n]}\rightarrow \R_{\geq 0}$ be the uniform distribution over the bases of $M$. By \cref{cor:jensonmatroid}, it would be enough to compute the marginals of $\mu$, but it can be seen that computing marginals is no easier than counting bases.

Instead, we use the convex programming framework described in \cref{sec:framework}. We claim that the optimum solution of the following concave program gives an additive $r$-approximation to $\H(\mu)=\log(|\B_M|)$ as well as a multiplicative $2$-approximation:
\begin{equation}
\label{cp:matroidbases}
	\tau=\max\set*{\sum_{i=1}^n \H(p_i)\given \p=(p_1,\dots,p_n)\in \P_M}.
\end{equation}
The objective function is a concave function of $\p$, so we can solve the above program using, e.g., the ellipsoid method.

\begin{proof}[Proof of \cref{thm:mainmatroid}]
	Let $\p = (p_1,\dots,p_n)\in \P_M$ be a vector achieving the maximum in \cref{cp:matroidbases}. The output of our algorithm will simply be
	\[ \ALG=e^\tau=\exp\parens*{\sum_{i=1}^n \H(p_i)}.\] 
	By \cref{prop:EntropyCounting}, the entropy $\H(\mu)$ equals $\log(\card{\B_M})$. Therefore to prove \cref{thm:mainmatroid}, it suffices to show that $\tau=\log(\ALG)$ is an additive $r$-approximation  and also a multiplicative 2-approximation of $\H(\mu)$, i.e., $\max\set*{\frac{1}{2}\tau, \tau-r}\leq \H(\mu)\leq \tau$.
	
	Firstly, note that since $(\mu_1,\dots,\mu_n)\in \P_M$, we have
	\begin{equation}\label{eq:upperboundsinglematroid} \tau\geq \sum_{i=1}^n \H(\mu_i)\geq \H(\mu), \end{equation}
	where the first inequality follows from the definition, \cref{cp:matroidbases}, and the second inequality follows from the subadditivity of entropy, \cref{fact:entropyupperbound}.

	Secondly, since $\p$ is in the polytope $\P_M = \P_{\mu}$, by \cref{lem:muCLogConv}, there is a probability distribution $\tmu$ on the bases of $M$ such that for all $i$, $\tmu_i = p_i$, and both $\tmu$ and $\tmu^*$ are log-concave. Applying \cref{cor:homadditiveentropy,cor:entropyfactor2}  to $\tmu$, we get
	\[ \H(\tmu)\geq \max\set*{\frac{1}{2}\sum_{i=1}^n \H(\tmu_i), \sum_{i=1}^n \H(\tmu_i)-r}. \]
	But note that $\sum_{i=1}^n\H(\tmu_i)=\sum_{i=1}^n\H(p_i)=\tau$, so
	$ \H(\tmu)\geq \max\set*{\frac{1}{2}\tau, \tau-r}$.
	
	Since $\mu$ is the uniform distribution over its support, and $\supp(\tmu)\subseteq \supp(\mu)$, its follows from \cref{prop:EntropyCounting} that $\H(\mu)\geq \H(\tmu)$. So we find that
	\[ \H(\mu)\geq \max\set*{\frac{1}{2}\tau, \tau-r}, \]
	which together with \cref{eq:upperboundsinglematroid} finishes the proof.
\end{proof}

	\section{Counting Common Bases of Two Matroids}
\label{sec:matroidintersection}

In this section we prove \cref{thm:matroidintersection}. Given two matroids $M,N$ on the ground set of elements $[n]$ we want to estimate the number of common bases of $M$ and $N$. We may assume, trivially, that both matroids are of the same rank $r$.  Following the same framework we used for a single matroid, described in \cref{sec:framework}, we solve the following concave program
\begin{equation}
	\label{eq:matroidintersectionCP}
	\tau = \max\set*{\sum_{i=1}^n \H(p_i)\given \p=(p_1,\dots,p_n)\in \P_M\cap \P_N},
\end{equation}
using, e.g., the ellipsoid method, and output $\ALG=e^\tau$ as our estimate for the number of common bases. Our main result is that $\ALG$ provides a multiplicative $2^{O(r)}$ approximation to $\card{\B_M\cap \B_N}$, or equivalently that $\tau$ provides an additive $O(r)$ approximation to $\log(\card{\B_M\cap \B_N})$.
\begin{proof}[Proof of \cref{thm:matroidintersection}]
	If $\mu$ is the uniform distribution on $\B_M\cap \B_N$, then the vector of its marginals $(\mu_1,\dots,\mu_n)$ belongs to $\P_M\cap \P_N$. It follows that
	\[ \tau\geq \sum_{i=1}^n \H(\mu_i)\geq \H(\mu) =\log(\card{\B_M\cap\B_N}), \]
	where the second inequality is an application of the subadditivity of entropy, \cref{fact:entropyupperbound}. To prove \cref{thm:matroidintersection}, it suffices to show that $\tau-O(r)$ is a lower bound for $\H(\mu)=\log(\card{\B_M\cap\B_N})$.
	
	Let $g_M(\y)\in \R[y_1,\dots,y_n]$ be the generating polynomial of $M$ and $g_{N^*}(\z)\in \R[z_1,\dots,z_n]$ be the generating polynomial of $N^*$, the dual matroid of $N$. Then the product $g_M(\y)g_{N^*}(\z)$ is the generating polynomial of the direct sum $M\oplus  N^*$, which is a matroid:
	\[ g_{M\oplus N^*}(\y, \z)=g_M(\y)g_{N^*}(\z). \]
Then by \cref{thm:gMlogconcave}, $g_M(\y)g_{N^*}(\z)$ is a completely log-concave polynomial. Since $\deg(g_M)=r$ and $\deg(g_{N^*})=n-r$, the product  $g_M(\y)g_{N^*}(\z)$ is a polynomial of degree $n$ in the $2n$ variables $y_1,\dots,y_n,z_1,\dots,z_n$. This polynomial fully encodes the matroids $M$ and $N$. In particular we can obtain $\card{\B_M\cap \B_N}$ as the following expression
	\begin{equation}
	\label{eq:intersectionmixedpoly}
		\card{\B_M\cap\B_N}=\parens*{\prod_{i=1}^n(\partial_{y_i}+\partial_{z_i})}g_M(\y)g_{N^*}(\z).
	\end{equation}
	To see this, observe that we can rewrite the right hand side as
	\[ \sum_{S\subseteq [n]}\parens*{\prod_{i\in S}\partial_{y_i}}\parens*{\prod_{i\in [n]\setminus S} \partial_{z_i}}g_M(\y)g_{N^*}(\z). \]
	The term corresponding to each $S$ is zero unless $S$ is independent in $M$ and $[n]\setminus S$ is independent in $N^*$. But this can only happen when $S$ is a common basis of $M$ and $N$, and in that case, this term is the constant $1$.
	
	We will use complete log-concavity of $g_M(\y)g_{N^*}(\z)$ to prove a lower bound on the expression in \cref{eq:intersectionmixedpoly}. Conveniently, the differential operators in \cref{eq:intersectionmixedpoly} are of the type $\partial_{\v}$ for $\v\in \R_{\geq 0}^n$, under which completely log-concave polynomials are closed, and this will be crucial for the proof. We will apply \cref{prop:joint}, which will be fully stated and proved later in this section, to the polynomial $g_M(\y)g_{N^*}(\z)$, to show that for any $\p=(p_1,\dots,p_n)\in [0,1]^n$, the following inequality holds (see \cref{cor:capacitylb}):
	\begin{equation}
	\label{eq:capacityp}
		\begin{aligned}
		\eval{\parens*{\prod_{i=1}^n(\partial_{y_i}+\partial_{z_i})}g_M(\y)g_{N^*}(\z)}_{\y=\z=0}&\geq \parens*{\frac{\p}{e^2}}^\p \cdot \inf_{\y,\z>0} \frac{g_M(\y)g_{N^*}(\z)}{\y^\p \z^{1-\p}}\\
		&=\parens*{\frac{\p}{e^2}}^\p\cdot\parens*{\inf_{\y>0}\frac{g_M(\y)}{\y^\p}}\cdot\parens*{\inf_{\z>0}\frac{g_{N^*}(\z)}{\z^{1-\p}}}.
		\end{aligned}
	\end{equation}
	Notice that evaluation at $\y=\z=0$ is inconsequential here, because the expression being evaluated is a constant. We will be particularly interested in the case where $\p \in \P_M\cap \P_N$. When this happens, $\p\in \P_N$, which implies that $1-\p\in \P_{N^*}$. By \cref{lem:muCLogConv}, there are distributions $\nu$ on $\B_M$ and $\omega$ on $\B_{N^*}$, obtained as limits of external fields applied to uniform distributions, with marginals $\p, 1-\p$, such that $\nu, \nu^*, \omega, \omega^*$ are all completely log-concave.
	\[ \forall i\in [n]: \nu_i=p_i\qquad\text{and}\qquad\forall i\in [n]: \omega_i=1-p_i. \]
	By \cref{lem:geometric-program}, we have
	\[ \H(\nu)=\log\parens*{\inf_{\y>0}\frac{g_M(\y)}{\y^\p}}\qquad\text{and}\qquad\H(\omega)=\log\parens*{\inf_{\z>0}\frac{g_{N^*}(\z)}{\z^{1-\p}}}. \]
	Now we apply \cref{thm:mainentropy} to $\nu, \omega^*$. Note that the marginals of both $\nu$ and $\omega^*$ are $\p$. So we get
	\[ \min\set*{\H(\nu), \H(\omega^*)}\geq \sum_{i=1}^n p_i\log\frac{1}{p_i}. \]
	Noting that $\H(\omega^*)=\H(\omega)$ and combining the previous two equations, we get
	\[ \log\parens*{\parens*{\inf_{\y>0}\frac{g_M(\y)}{\y^\p}}\cdot\parens*{\inf_{\z>0}\frac{g_{N^*}(\z)}{\z^{1-\p}}}}\geq 2\sum_{i=1}^n p_i\log \frac{1}{p_i}. \]
	Plugging this back into \cref{eq:capacityp}, and using \cref{eq:intersectionmixedpoly}, we get
	\begin{align*}
		\log\parens*{\card{\B_M\cap\B_N}}&\geq \sum_{i=1}^n p_i\log \frac{p_i}{e^2}+2\sum_{i=1}^n p_i\log \frac{1}{p_i}\\
		&=\sum_{i=1}^n p_i\log \frac{1}{p_i}-2\sum_{i=1}^n p_i\\
		&=\sum_{i=1}^n \H(p_i)-\sum_{i=1}^n (1-p_i)\log\frac{1}{1-p_i}-2\sum_{i=1}^n p_i\\
		&\geq \sum_{i=1}^n \H(p_i)-3r,
	\end{align*}
	where, for the last inequality we use the fact that $(1-p_i)\log\frac{1}{1-p_i}\leq p_i$ for $0\leq p_i\leq 1$ and $\sum_{i=1}^n p_i=r$. In summary, for any $\p\in \P_M\cap \P_N$, the quantity $\sum_{i=1}^n \H(p_i)-3r$ is a lower bound for $\log(\card{\B_M\cap\B_N})$. This completes the proof.
\end{proof}
Now we state and prove the technical part used in the proof above. 
\begin{theorem}
\label{prop:joint}
Let $g\in\R[y_1,\dots,y_n,z_1,\dots,z_n]$ be a completely log-concave multiaffine polynomial and $\p\in [0, 1]^n$. Then,
\[ \eval{\parens*{\prod_{i=1}^n (\partial_{y_i}+\partial_{z_i})}g(\y, \z)}_{\y=\z=0}\geq \phi(\p) \cdot \inf_{\y, \z\in \R_{>0}^n} \frac{g(\y,\z)}{\y^\p \z^{1-\p}}, \]
where $\phi(\p)$ is independent of the polynomial $g$, and is given by the following expression:
\[ \phi(\p)=\prod_{i=1}^n \parens*{p_i^{p_i}\cdot (1-p_i)^{1-p_i}\cdot \frac{1}{1+p_i(1-p_i)}}. \]
\end{theorem}
Before proving \cref{prop:joint}, note that we can simplify it, by using the following inequalities which are valid for all $p_i\in [0, 1]$:
\[ (1-p_i)^{1-p_i}\geq e^{-p_i}\qquad\text{and}\qquad \frac{1}{1+p_i(1-p_i)}\geq e^{-p_i}. \]
Both inequalities can be proven by taking $\log$ from the left hand sides to get convex functions of $p_i$, and lower bounding by the tangent line at $p_i=0$. Together they imply that
\[ \phi(\p)\geq \prod_{i=1}^n \parens*{p_i^{p_i}\cdot e^{-p_i}\cdot e^{-p_i}}=\parens*{\frac{\p}{e^2}}^\p, \]
giving the following. 
\begin{corollary}
\label{cor:capacitylb}
	For any completely log-concave multiaffine polynomial $g\in \R[y_1,\dots,y_n,z_1,\dots,z_n]$ and $\p\in [0,1]^n$ the following inequality holds:
	\[ \eval{\parens*{\prod_{i=1}^n (\partial_{y_i}+\partial_{z_i})}g(\y, \z)}_{\y=\z=0}\geq \parens*{\frac{\p}{e^2}}^\p \inf_{\y, \z\in \R_{>0}^n} \frac{g(\y,\z)}{\y^\p \z^{1-\p}}. \]
\end{corollary}

\begin{remark}
We remark that the statement and the proof of \cref{prop:joint} are inspired by those of a similar statement in an earlier work of a subset of authors \cite{AO17}, involving real stable polynomials. The quantity $\inf_{\y,\z>0}g(\y,\z)/\y^\p \z^{1-\p}$ and similar expressions called the ``capacity'' of polynomials have been studied in several works, starting with \textcite{Gur06, Gur08, Gur09} and some more recent works \cite{NS16, AO17, SV17, Lea18}.
\end{remark}

In the rest of this section, we prove \cref{prop:joint}. The proof is by induction on $n$, and both the base case and the induction step reduce to the case of $n=1$. So we will prove this first. An important step is identifying \emph{bivariate} completely log-concave multiaffine polynomials $g(y,z)\in\R[y,z]$. 

\begin{lemma}
	\label{lem:bivariate}
	A polynomial $g(y,z)= a+by+cz+dyz\in \R[y,z]$ with nonnegative coefficients is completely log-concave if and only if  $2bc\geq ad$. Moreover, if $g$ is completely log-concave and $d=1$ then $ (b+z)(c+y) + bc\geq g(y,z)$ for all $(y,z)\in\R^2$. 
\end{lemma}
\begin{proof}
	Since $g$ has degree two, it is completely log-concave if and only if it is log-concave, which happens if and only if the matrix 
	\[ g^2\cdot \Hess{\log g}= \begin{bmatrix}-(b+dz)^2&ad-bc\\ ad-bc &-(c+dz)^2 \end{bmatrix} \]
	is negative semidefinite for all $y,z\in \R_{\geq0}$. Its determinant simplifies to 
	\[\det\parens*{g^2\cdot \Hess{\log g}} = d \cdot g(y,z) \cdot (2 b c - a d + b d y + c d z + d^2 y z). \]
	Note that the first two factors are nonnegative on $\R_{\geq 0}^2$, and the last factor is nonnegative on $\R_{\geq 0}^2$ if and only if it is nonnegative at $(y,z) = (0,0)$. 
	Since the diagonal elements of the $2\times 2$ matrix are already nonpositive, this implies that 
	$g$ is  log-concave if and only if $2bc - ad \geq 0$.

	If $d = 1$, then $2bc \geq a$ and for all $(y,z)\in \R^2$, $ (b+z)(c+y) + bc-  g(y,z) = 2bc -a \geq 0$.
\end{proof}

Now we state and prove the $n=1$ case of \cref{prop:joint}.
\begin{lemma}\label{lem:n=1Prop}
	If the polynomial $g(y,z)=a+by+cz+dyz\in \R[y,z]$ is log-concave, then for any $p\in [0,1]$,
	\[ b+c \geq \phi(p) \cdot \inf_{y,z>0}  \frac{g(y,z)}{y^p\cdot z^{1-p}}\qquad\text{where}\qquad \phi(p) = \frac{ p^p (1-p)^{1-p}}{1+p(1-p)}. 
\]
\end{lemma}
\begin{proof}
	First let us resolve the boundary case $p\in\set{0,1}$. In this case, $\phi(p)=1$. For $p = 0$, taking $z\to \infty$ and $y\to 0$, shows the desired inequality. Similarly, if $p=1$, we may take $y\to \infty$ and $z\to 0$. So from now on, assume that $p\notin\set{0,1}$.

	Now we make some simplifying assumptions on $g$ and deal with the case $d = 0$. If additionally $c=0$, then $g(y,z)$ does not involve the variable $z$, meaning that the infimum above is zero and the inequality is satisfied. A similar argument resolves the case $b=0$. Finally, if $bc \neq 0$,  then by \cref{lem:bivariate}, for sufficiently small $\epsilon >0$, $g(y,z)+\epsilon yz$ is also log-concave. Therefore we can assume $d$ is nonzero and rescale $g(y,z)$ so that $d = 1$. Using the second part of \cref{lem:bivariate}, it then suffices to take $g(y,z) = (b+z)(c+y) + bc$. 

	Consider $\haty=c \cdot \frac{p}{1-p}$ and $\hatz=b\cdot \frac{1-p}{p}$. For this choice of $\haty,\hatz$, we compute that 
	\[ g(\haty,\hatz) = \parens*{b + b\cdot \frac{1-p}{p}}\parens*{c +  c\cdot \frac{p}{1-p}} +bc = bc \cdot \parens*{\frac{1}{p}  \cdot \frac{1}{1-p}  + 1} = bc \cdot  \frac{ 1 + p(1-p) }{p(1-p)}. \]
	The function of interest then evaluates to  
	\[ \frac{g(\haty,\hatz)}{\haty^p\hatz^{1-p}} = bc \cdot  \frac{ 1 + p(1-p) }{p(1-p)} \cdot \frac{(1- p)^p}{c^p p^p} \cdot \frac{ p^{1-p}}{b^{1-p} (1-p)^{1-p}}= b^p c^{1-p} \cdot  \frac{ 1 + p(1-p) }{\parens*{p^p(1-p)^{1-p}}^2}, \]
	and 
	\[ \phi(p) \cdot \frac{g(\haty,\hatz)}{\haty^p\hatz^{1-p}}=\frac{ b^p c^{1-p} }{p^p(1-p)^{1-p}}\leq p\cdot \frac{b}{p}+(1-p)\cdot \frac{c}{1-p}=b+c. \]
	Here the inequality follows from the weighted arithmetic mean geometric mean inequality. 
\end{proof}

We are now ready to prove \cref{prop:joint}.

\begin{proof}[Proof of \cref{prop:joint}]
	We will proceed by induction on $n$. The base case $n=1$ is the content of \cref{lem:n=1Prop}. 
	Suppose that the proposition holds  in $\R[y_1,\dots,y_n, z_1,\dots,z_n]$ and let $g(s,\y,t,\z)$ be a completely log-concave multiaffine polynomial in $ \R[s,\y,t,\z]$. Let us define the polynomial 
	\[ h(\y, \z) = \eval{(\partial_{s}+\partial_{t})g}_{s=t=0}\in \R[\y, \z].\]

	Since complete log-concavity is preserved under differentiation and restriction, $h(\y, \z)$ is also completely log-concave, and by induction, for any $\p\in [0,1]^n$ and any $\epsilon >0$, there exist $\hatbmy, \hatbmz \in \R_{>0}^n$ for which
	\[ \eval{\parens*{(\partial_s + \partial_t)\prod_{i=1}^{n} (\partial_{y_i}+\partial_{z_i})} g}_{s=t=\y=\z=0}=\eval{\parens*{\prod_{i=1}^{n} (\partial_{y_i}+\partial_{z_i})}h}_{\y=\z=0}\geq \phi(\p)\cdot  \frac{h(\hatbmy,\hatbmz)}{\hatbmy^{\p}\hatbmz^{1-\p}} - \epsilon.\]

	Now consider the bivariate polynomial $f\in \R[s,t]$ given by $f(s,t)=g(s,\hatbmy,t,\hatbmz)$, which is also completely log-concave by \cref{lem:CLCpreservers}. By \cref{lem:n=1Prop}, for any $q\in [0,1]$ and any $\delta>0$, there exist $\hats, \hatt>0$ for which 
	\[\eval{(\partial_s + \partial_t)f}_{s=t=0}\geq \phi(q)\cdot \frac{g(\hats,\hatt)}{\hats^q \hatt^{1-q} } - \delta. \]
	Notice that $\eval{(\partial_s + \partial_t)f}_{s=t=0}$ equals $h(\hatbmy, \hatbmz)$. Putting this all together, we find that 
	\begin{align*}
		\eval{\parens*{(\partial_s + \partial_t)\prod_{i=1}^{n} (\partial_{y_i}+\partial_{z_i})} g}_{s=t=\y=\z=0} & \geq \phi(\p)\cdot \frac{h(\hatbmy,\hatbmz)}{\hatbmy^{\p}\hatbmz^{1-\p}} - \epsilon  \\
		&  =\frac{\phi(\p)}{\hatbmy^\p\hatbmz^{1-\p}}\cdot \parens*{\eval{(\partial_{s}+\partial_{t})f}_{s=t=0}} - \epsilon\\
		& \geq \frac{\phi(\p)}{\hatbmy^{\p}\hatbmz^{1-\p}}\cdot \parens*{\frac{\phi(q)}{\hats^{q}\hatt^{1-q}}\cdot f(\hats, \hatt) - \delta} -\epsilon.
	\end{align*}
	Since $f(\hats, \hatt) = g(\hats, \hatbmy, \hatt, \hatbmz)$ and
	\[
 		\frac{\phi(\p)}{\hatbmy^\p\hatbmz^{1-\p}}\cdot \frac{\phi(q)}{\hats^q\hatt^{1-q}}\cdot f(\hats, \hatt) = \frac{\phi(q,\p)}{(\hats,\hatbmy)^{(q,\p)} (\hatt,\hatbmz)^{1-(q,\p)}}\cdot g(\hats,\hatbmy,\hatt,\hatbmz),
	\]
	taking $\epsilon, \delta \to 0$ at a rate that ensures $\delta \cdot \phi(\p)/\hatbmy^\p \hatbmz^{1-\p}\to 0$ proves that the proposition holds  in the polynomial ring $ \R[s,\y, t, \z]$.
\end{proof}

	\section{Weighted Counts of Common Matroid Bases}
\label{sec:weighted}

In this section we prove \cref{thm:weighted}. As mentioned in \cref{sec:framework}, the algorithm we use is just a simple modification of the unweighted case. For given matroids $M$ and $N$ of rank $r$ on the ground set $[n]$, and weights $\bmlambda=(\lambda_1,\dots,\lambda_n)\in \R_{\geq 0}^n$, we simply solve the following concave program
\[ \tau = \max\set*{\sum_{i=1}^n \parens*{p_i\log \frac{\lambda_i}{p_i}+(1-p_i)\log \frac{1}{1-p_i}}\given \p \in \P_M\cap \P_N}, \]
using, e.g., the ellipsoid method, and output $\ALG=e^\tau$ as an approximation to the $\bmlambda$-weighted count of $\B_M\cap \B_N$, namely $\sum_{B\in \B_M\cap \B_N} \bmlambda^B$.  

It is possible to give a direct proof of \cref{thm:weighted} by carrying the weights through the proofs of \cref{thm:mainmatroid,thm:matroidintersection}. Here, we provide an alternative proof by reduction to the unweighted case.

First let us prove the easy side, that $\ALG$ is an upper bound for the logarithm of the $\bmlambda$-weighted count. As in the unweighted case, we prove this for arbitrary subsets of $2^{[n]}$, not just common bases of matroids.
\begin{lemma}
\label{lem:weighted-upperbound}
	Suppose that $\B\subseteq 2^{[n]}$ and define $\P_\B=\conv\set{\1_B\given B\in \B}$. Let $\bmlambda=(\lambda_1,\dots,\lambda_n)\in \R_{\geq 0}^n$ and let $\mu$ be the uniform distribution on $\B$. Let $\p=(p_1,\dots,p_n)\in \P_\B$ be the marginals of the probability distribution $\ef{\bmlambda}{\mu}$. Then
	\[ \sum_{i=1}^n\parens*{p_i \log \frac{\lambda_i}{p_i}+(1-p_i)\log\frac{1}{1-p_i} }\geq \log\parens*{\sum_{B\in \B}\bmlambda^B}. \]
\end{lemma}
\begin{proof}
	Let $Z$ denote the $\bmlambda$-weighted count of $\B$, namely $\sum_{B\in \B}\bmlambda^B$. Note that the probability distribution $\ef{\bmlambda}{\mu}$ is supported on $\B$ and for $S\in \B$, is given by
	\[ \PrX{\ef{\bmlambda}{\mu}}{S}=\frac{\bmlambda^S}{\sum_{B\in \B} \bmlambda^B}=\frac{\bmlambda^S}{Z}. \]
	Applying the subadditivity of entropy, \cref{fact:subadditivity}, to the probability distribution $\ef{\bmlambda}{\mu}$ we get
	\begin{align*}
		\sum_{i=1}^n\H(p_i)&\geq \H(\ef{\bmlambda}{\mu})=\sum_{S\in \B} \frac{\bmlambda^S}{Z}\log \frac{Z}{\bmlambda^S}= \log(Z)+\sum_{S\in \B} \frac{\bmlambda^S}{Z}\parens*{\sum_{i\in S}\log \frac{1}{\lambda_i}}\\
		&=\log(Z)+\sum_{i=1}^n\parens*{\log \frac{1}{\lambda_i}\cdot \sum_{S\in \B:i\in S}\frac{\bmlambda^S}{Z}}=\log(Z)+\sum_{i=1}^n p_i \log \frac{1}{\lambda_i}. 
	\end{align*}
	Rearranging yields the desired statement.
\end{proof}
Now we prove \cref{thm:weighted}.
\begin{proof}[Proof of \cref{thm:weighted}]
	One side of the desired inequality is an immediate consequence of \cref{lem:weighted-upperbound}, when we let $\B=\B_M\cap \B_N$. So we just need to prove the other side, i.e.,
	\[ 2^{-O(r)}\ALG \leq \sum_{B\in \B_M\cap\B_N}\bmlambda^B, \]
	or equivalently, we will show that for any $\p\in \P_M\cap \P_N$
	\[ \sum_{i=1}^n p_i\log \frac{\lambda_i}{p_i}-O(r)\leq \log\parens*{\sum_{B\in \B_M\cap\B_N}\bmlambda^B}. \]
	Notice that we have dropped the terms $\sum_{i=1}^n (1-p_i)\log\frac{1}{1-p_i}$, because they are bounded by $O(r)$.
	
	For $\bmlambda=\1_{[n]}$, this is the content of \cref{thm:matroidintersection}. Our strategy is to prove the result first for $\bmlambda\in \Z_{\geq 0}^n$; we will then prove the case $\bmlambda\in \Q_{\geq 0}^n$. The result would follow by continuity, because for any fixed $\p$, both sides of the inequality are continuous functions of $\bmlambda$, and $\Q_{\geq 0}^n$ is dense in $\R_{\geq 0}^n$.
	
	Suppose that $\bmlambda\in \Z_{\geq 0}^n$. Let us introduce two new matroids, $\tM, \tN$, to which we will apply \cref{thm:matroidintersection}. We obtain these matroids, by replacing element $i\in [n]$ of the ground set with $\lambda_i$ new parallel elements $i^{(1)}, \dots, i^{(\lambda_i)}$; when $\lambda_i=0$, this corresponds to deleting the element $i$. Let $\pi:\set{i^{(j)}\given i\in [n], j\in [\lambda_i]}\to [n]$ be the projection map from the new ground set to the old ground set given by the equation $\pi(i^{(j)})=i$.
	A subset $I\subseteq \set{i^{(j)}\given i\in [n], j\in [\lambda_i]}$ is independent in $\tM$ or $\tN$ if and only if it contains no two parallel elements $i^{(j)}$ and $i^{(k)}$ for $j\neq k$ and its projection $\pi(I)$ is independent in $M$ or $N$, respectively. The key observation is that
	\[ \card{\B_\tM\cap \B_\tN}=\sum_{B\in \B_M\cap \B_N} \bmlambda^B. \]
	This is because for any common basis $\tB$ of $\tM, \tN$, the projection $B=\pi(\tB)$ is a common basis of $M$ and $N$; furthermore, for any choice of common basis $B\in \B_M\cap \B_N$, there are $\prod_{i\in B}\lambda_i$ choices for $\tB$ containing no pair of parallel elements, such that the projection $\pi(\tB)=B$; any such $\tB$ is a common basis of $\tM, \tN$.
	
	Let $\p\in \P_M\cap \P_N$. It follows that there is a distribution $\mu$ supported on $\B_M\cap \B_N$ whose marginals are $\p$. From $\mu$, we define a distribution $\tmu$ supported on $\B_\tM\cap \B_\tN$ 
	by sampling $B=\set{i_1,\dots,i_r}\in \B_M\cap \B_N$ according to $\mu$ and letting $\tB=\set{i_1^{(j_1)},\dots, i_r^{(j_r)}}$, where each $j_k$ is sampled uniformly at random from $[\lambda_k]$. 
	If $\tmu$ is the resulting distribution of $\tB$, then for any common basis $\tB$ of $\tM, \tN$ with projection $\pi(\tB) = B$, this gives $\tmu(\tB) = \mu(B)\prod_{i\in B}\lambda_i^{-1}$. 
	The marginals of $\tmu$ are given by
	\[ \tmu_{i^{(j)}}=\frac{p_i}{\lambda_i}. \]
	From \cref{thm:matroidintersection} it follows that
	\begin{align*} 
		\log(\card{\B_\tM\cap \B_\tN})&\geq \sum_{i=1}^n\sum_{j=1}^{\lambda_i} \tmu_{i^{(j)}}\log \frac{1}{\tmu_{i^{(j)}}}-O(r)\\ 
		&=\sum_{i=1}^n\sum_{j=1}^{\lambda_i} \frac{p_i}{\lambda_i}\log\frac{\lambda_i}{p_i}-O(r) \\
		&= \sum_{i=1}^n p_i \log \frac{\lambda_i}{p_i}-O(r),
	\end{align*}
	which completes the proof for $\bmlambda\in \Z_{\geq 0}^n$.
	
	Now let us prove this for $\bmlambda\in \Q_{\geq 0}^n$. Let $\p\in \P_M\cap\P_N$ and let $t$ be such that $t\bmlambda\in \Z_{\geq 0}^n$. Then	\[ \log\parens*{\sum_{B\in \B_M\cap \B_N} (t\bmlambda)^B}\geq \sum_{i=1}^n p_i\log \frac{t\lambda_i}{p_i}-O(r). \]
	Notice that the expression inside the $\log$ on the left hand side is homogeneous of degree $r$ in $t$. For the right hand side,  observe that $\sum_{i=1}^n p_i=r$, so we can simplify the inequality to
	\[ r\log(t)+\log\parens*{\sum_{B\in \B_M\cap \B_N}\bmlambda^B}\geq r\log(t)+\sum_{i=1}^n p_i\log \frac{\lambda_i}{p_i}-O(r), \]
	which finishes the proof for $\bmlambda\in \Q_{\geq 0}^n$ and consequently for all of $\R_{\geq 0}^n$.
\end{proof}
	
	\printbibliography
\end{document}